\documentclass[letterpaper,onecolumn,11pt,accepted=2022-10-23]{quantumarticle}
\pdfoutput=1
\usepackage[english]{babel}

\usepackage[T1]{fontenc}

\usepackage[margin=1in]{geometry}

\usepackage{amssymb,amsmath,amsthm,amsfonts}
\usepackage{textgreek}
\usepackage{mathtools}
\usepackage{enumitem}
\usepackage[numbers,comma,sort&compress]{natbib}
\usepackage{authblk}
\usepackage{graphicx}
\usepackage[font=small]{caption}
\usepackage[labelformat=simple]{subcaption}
\usepackage{float}
\usepackage[ruled,vlined,linesnumbered]{algorithm2e}
\usepackage{physics}
\usepackage{footnote}
\usepackage{xcolor}
\usepackage{mathrsfs}
\usepackage{bbm}

\usepackage[colorlinks]{hyperref}

\newtheorem{theorem}{Theorem}
\newtheorem{definition}{Definition}
\newtheorem{lemma}{Lemma}
\newtheorem{proposition}{Proposition}

\newtheorem{corollary}{Corollary}

\newcommand{\eqn}[1]{(\ref{eqn:#1})}
\newcommand{\eq}[1]{(\ref{eq:#1})}

\newcommand{\thm}[1]{\hyperref[thm:#1]{Theorem~\ref*{thm:#1}}}
\newcommand{\cor}[1]{\hyperref[cor:#1]{Corollary~\ref*{cor:#1}}}
\newcommand{\defn}[1]{\hyperref[defn:#1]{Definition~\ref*{defn:#1}}}
\newcommand{\lem}[1]{\hyperref[lem:#1]{Lemma~\ref*{lem:#1}}}
\newcommand{\prop}[1]{\hyperref[prop:#1]{Proposition~\ref*{prop:#1}}}
\newcommand{\fig}[1]{\hyperref[fig:#1]{Figure~\ref*{fig:#1}}}
\newcommand{\tab}[1]{\hyperref[tab:#1]{Table~\ref*{tab:#1}}}
\newcommand{\algo}[1]{\hyperref[algo:#1]{Algorithm~\ref*{algo:#1}}}
\renewcommand{\sec}[1]{\hyperref[sec:#1]{Section~\ref*{sec:#1}}}
\newcommand{\append}[1]{\hyperref[append:#1]{Appendix~\ref*{append:#1}}}
\newcommand{\fac}[1]{\hyperref[fac:#1]{Fact~\ref*{fac:#1}}}
\newcommand{\lin}[1]{\hyperref[lin:#1]{Line~\ref*{lin:#1}}}
\newcommand{\fnote}[1]{\hyperref[fnote:#1]{Footnote~\ref*{fnote:#1}}}

\newcommand{\rmk}[1]{\hyperref[rmk:#1]{Remark~\ref*{rmk:#1}}}

\def\>{\rangle}
\def\<{\langle}

\newcommand{\vect}[1]{\ensuremath{\mathbf{#1}}}
\newcommand{\x}{\ensuremath{\mathbf{x}}}

\newcommand{\N}{\mathbb{N}}
\newcommand{\Z}{\mathbb{Z}}
\newcommand{\R}{\mathbb{R}}
\newcommand{\C}{\mathbb{C}}

\DeclareMathOperator{\poly}{poly}

\DeclareMathOperator{\inv}{inv}
\DeclareMathOperator{\nrm}{norm}
\DeclareMathOperator{\diag}{diag}

\DeclareMathOperator{\hamt}{HAM-T}

\renewcommand{\d}{\mathrm{d}}
\renewcommand{\emptyset}{\varnothing}
\def\Tr{\operatorname{Tr}}\def\:{\hbox{\bf:}}
\newcommand{\range}[1]{[#1]}
\newcommand{\rangez}[1]{[{#1}]_0}

\let\oldnl\nl
\newcommand{\nonl}{\renewcommand{\nl}{\let\nl\oldnl}}

\begin{document}

\title{Quantum Simulation of Real-Space Dynamics}

\author{Andrew M. Childs}
\email{amchilds@umd.edu}
\affiliation{Joint Center for Quantum Information and Computer Science, University of Maryland}
\affiliation{Department of Computer Science, University of Maryland}
\author{Jiaqi Leng}
\email{jiaqil@umd.edu}
\affiliation{Joint Center for Quantum Information and Computer Science, University of Maryland}
\affiliation{Department of Mathematics, University of Maryland}
\author{Tongyang Li}
\email{tongyangli@pku.edu.cn}
\affiliation{Center on Frontiers of Computing Studies, Peking University} 
\affiliation{School of Computer Science, Peking University} 
\affiliation{Center for Theoretical Physics, Massachusetts Institute of Technology} 
\author{Jin-Peng Liu}
\email{jliu1219@terpmail.umd.edu}
\affiliation{Joint Center for Quantum Information and Computer Science, University of Maryland}
\affiliation{Department of Mathematics, University of Maryland}
\author{Chenyi Zhang}
\email{chenyiz@stanford.edu}
\affiliation{Institute for Interdisciplinary Information Sciences, Tsinghua University}


\maketitle

\begin{abstract}
Quantum simulation is a prominent application of quantum computers. While there is extensive previous work on simulating finite-dimensional systems, less is known about quantum algorithms for real-space dynamics. We conduct a systematic study of such algorithms. In particular, we show that the dynamics of a $d$-dimensional Schr{\"o}dinger equation with $\eta$ particles can be simulated with gate complexity\footnote{The $\tilde{O}$ notation omits poly-logarithmic terms. Specifically, $\tilde{O}(g)=O(g\poly(\log g))$.\label{fnote:tilde-O}} $\widetilde O\bigl(\eta d F \poly(\log(g'/\epsilon))\bigr)$, where $\epsilon$ is the discretization error, $g'$ controls the higher-order derivatives of the wave function, and $F$ measures the time-integrated strength of the potential.
Compared to the best previous results, this exponentially improves the dependence on $\epsilon$ and $g'$ from $\poly(g'/\epsilon)$ to $\poly(\log(g'/\epsilon))$ and polynomially improves the dependence on $T$ and $d$, while maintaining best known performance with respect to $\eta$.
For the case of Coulomb interactions, we give an algorithm using $\eta^{3}(d+\eta)T\poly(\log(\eta dTg'/(\Delta\epsilon)))/\Delta$ one- and two-qubit gates, and another using $\eta^{3}(4d)^{d/2}T\poly(\log(\eta dTg'/(\Delta\epsilon)))/\Delta$ one- and two-qubit gates and QRAM operations, where $T$ is the evolution time and the parameter $\Delta$ regulates the unbounded Coulomb interaction. We give applications to several computational problems, including faster real-space simulation of quantum chemistry, rigorous analysis of discretization error for simulation of a uniform electron gas, and a quadratic improvement to a quantum algorithm for escaping saddle points in nonconvex optimization.
\end{abstract}

\newpage
\section{Introduction}\label{sec:intro}

Simulating quantum physics is one of the primary applications of quantum computers~\cite{feynman1982simulating}. The first explicit quantum simulation algorithm was proposed by Lloyd~\cite{lloyd1996universal} using product formulas, and numerous quantum algorithms for quantum simulations have been extensively developed since then~\cite{kassal2008polynomial,kivlichan2017bounding,toloui2013quantum,babbush2017exponentially,babbush2019quantum, su2021fault,aspuru2005simulated,whitfield2011simulation,seeley2012bravyi,wecker2014gate,hastings2014improving,poulin2014trotter,mcclean2014exploiting,babbush2015chemical, reiher2017elucidating,babbush2016exponentially,motta2018low,campbell2019random,berry2019qubitization,childs2019theory,von2020quantum,lee2020even,su2020nearly,an2021time,jin2021quantum}, with various applications ranging from quantum field theory~\cite{jordan2012quantum,preskill2019simulating} to quantum chemistry~\cite{babbush2015chemical,bauer2020quantum,cao2019quantum} and condensed matter physics~\cite{babbush2017low,childs2018toward}.

The Schr{\"o}dinger equation that determines the evolution of a quantum wave function in $d$-dimensional real space has the form
	\begin{align}\label{eqn:Schrodinger-indep}
		i\frac{\partial}{\partial t}\Phi(\vect{x},t)=\Big[-\frac{1}{2}\nabla^{2}+f(\vect{x})\Big]\Phi(\vect{x},t)
	\end{align}
where $f\colon\R^{d}\to\R$ is the potential function.\footnote{More generally, we can consider time-dependent potentials as formulated in \eqn{Schrodinger-dep}.}

In this paper, we consider quantum simulations for general potential functions, which we model by assuming quantum oracle access to $f$. Specifically, we assume a unitary $U_f$ such that for any $\x\in\R^{d}$ and $z\in\R$,
	\begin{align}\label{eqn:evaluation-indep}
		U_{f}|\x\>|z\>=|\x\>|f(\x)+z\>.
	\end{align}
In practice, real numbers used in the simulation will be represented digitally, but we assume the representation has sufficiently high precision that errors from this digital representation can be neglected. This model allows coherent superpositions of queries to the potential function $f$, which is a standard assumption for quantum algorithms working in real space, including quantum simulation~\cite{kivlichan2017bounding} and optimization~\cite{vanApeldoorn2020optimization,chakrabarti2020optimization,zhang2020quantum} algorithms. Note that if $f$ can be computed by a classical circuit, then the corresponding quantum oracle can be implemented by a quantum circuit of roughly the same size.

The first work on real-space quantum simulation algorithms dates back to Wiesner~\cite{wiesner1996simulations} and Zalka~\cite{zalka1998efficient}, who used product formulas to simulate the time evolution by separately handling the kinetic and potential terms, relating them with the quantum Fourier transform. More recently, Kassal et al.~\cite{kassal2008polynomial} developed a real-space simulation algorithm for chemical dynamics using a different approach. They concluded that simulating dynamics in real space can be more accurate and efficient than a second-quantized approach using the Born-Oppenheimer approximation.

To simulate real-space dynamics on a digital computer, we must discretize the spatial degrees of freedom.
Although~\cite{wiesner1996simulations,zalka1998efficient,kassal2008polynomial} estimated the gate complexity of their quantum simulation algorithms, these early-stage results did not rigorously analyze how the complexity depends on discretization error. As far as we are aware, the first complexity analysis including the discretization error was conducted by Kivlichan et al.~\cite{kivlichan2017bounding}, which developed a quantum algorithm for simulating real-space dynamics using high-order finite difference schemes and Hamiltonian simulation with a truncated Taylor series. Their algorithm has worst-case complexity $\widetilde O(\exp(\eta d))$ assuming a bounded potential, or $\widetilde O\bigl(\eta^7d^4T^3k_{\max}^2/\epsilon^2)$ given a strong assumption about the derivatives of the wave function, where $\eta$ is the number of particles, $d$ is the dimension, $T$ is the evolution time, $\epsilon$ is the discretization error, and $k_{\max}$ (defined in~\cite[Corollary 5]{kivlichan2017bounding}) controls the higher-order derivatives of the wave function. The exponential scaling of the former result arises from possible singularities in the wave function (in particular, it results from an upper bound on the integration error of $d\eta$-dimensional wave functions in \cite[Theorem 4]{kivlichan2017bounding}).
We similarly assume the wave function is sufficiently regular by introducing a related parameter $g'$ (defined in \eqn{gprime}).

Real-space quantum simulation is a form of first-quantized quantum simulation. First quantization represents the overall quantum state by storing the location of each particle, whereas second quantizaton describes the occupation numbers of all possible locations. Previous work has studied the complexity of first-quantized simulations using various basis sets such as Gaussian orbitals~\cite{toloui2013quantum,babbush2017exponentially} and plane waves~\cite{babbush2019quantum}. While first-quantized simulation using plane waves is similar to real-space simulation in the Fourier basis (as considered in this paper), the main difference is that the former methods choose a fixed number $N$ of basis functions in the Galerkin representation, rather than aiming to approximate the underlying real-space dynamics within a given allowed error using the pseudospectral representation.\footnote{The Galerkin and pseudospectral representations are introduced and compared in \sec{fourier-spectral}.}
Recently, Su et al.~\cite{su2021fault} introduced a first-quantized quantum simulation algorithm that considered a real-space grid representation as in the work of Kassal et al.~\cite{kassal2008polynomial}, but employing qubitization~\cite{low2019qubitization} and interaction picture~\cite{low2018hamiltonian} techniques to achieve upper bounds of $\widetilde O(\eta^{8/3}N^{1/3}T+\eta^{4/3}N^{2/3}T)$ and $\widetilde O(\eta^{8/3}N^{1/3}T)$, respectively, where $N$ is the number of grid points. The latter complexity matches the best known scaling of first-quantized methods~\cite{babbush2019quantum}. Compared to the real-space quantum simulation result of Kivlichan et al.~\cite{kivlichan2017bounding}, Su et al.~focused more on the $N$-dependence of the complexity than on other parameters.
Appendix K of \cite{su2021fault} indicates that the factor of $N^{1/3}$ results from an upper bound on the potential term in Eq.~(K7), but further work is needed to better understand the required dependence of $N$ on $T$, $\epsilon$, and $g'$.

Many quantum algorithms for simulating quantum chemistry rely on second quantization. In particular, algorithms for the electronic structure problem using a second-quantized representation are widely studied as a near-term application of quantum computers \cite{aspuru2005simulated}. Work on this topic has adopted different representations including Gaussian orbitals~\cite{aspuru2005simulated,whitfield2011simulation,seeley2012bravyi,wecker2014gate,hastings2014improving,poulin2014trotter,mcclean2014exploiting,babbush2015chemical,reiher2017elucidating,babbush2016exponentially} and plane waves~\cite{babbush2017low,low2018hamiltonian,berry2019qubitization,childs2019theory,su2020nearly} in search of algorithms with lower resource requirements.

Although second-quantized approaches to quantum simulation are perhaps more widely studied, there is growing interest in first quantization. In particular, the aforementioned work of Su et al. \cite{su2021fault} recently gave a systematic study of the practical performance of first-quantized simulation methods. While the worst-case complexity of simulating first-quantized real-space dynamics with a bounded potential scales as $\widetilde O(\exp(\eta d))$~\cite{kivlichan2017bounding}, first-quantized simulation enjoys asymptotically lower space and gate complexity in terms of $\eta$ and $N$ when considering algorithms that work with a fixed number of basis functions $N$, as mentioned above. For simulating arbitrary-basis electronic structure Hamiltonians, the space complexities of these general-purpose first- and second-quantized algorithms are $\widetilde O(\eta\log N)$ and $\widetilde O(N)$, respectively, and the best gate complexities we are aware of are $\widetilde O(\eta^{\frac{8}{3}}N^{\frac{1}{3}})$~\cite{babbush2019quantum,su2021fault} and $\widetilde O(N^5)$~\cite{babbush2016exponentially}, respectively.\footnote{Some studies suggest that the gate complexity of second-quantized algorithms could grow more slowly with $N$ for certain models and representations~\cite{berry2019qubitization,lee2020even,su2020nearly}.} Since $N=\Omega(\eta)$ due to the Pauli exclusion principle, the space and gate complexities of second-quantization algorithms are no better than those of first-quantization algorithms. Furthermore, first-quantized simulation can simulate the full dynamics of molecular Schr{\"o}dinger equations, while second-quantized simulation usually operates in the Born-Oppenheimer approximation with electronic orbitals chosen for fixed nuclear positions. In addition, the choice of basis functions for second-quantized simulation algorithms can depend heavily on prior knowledge. The complexity of simulating Hamiltonian systems using heuristic basis sets has been well analyzed, but the discretization error from the original continuum system has only been discussed asymptotically, and in many studies is simply neglected. As pointed out by Kassal et al.~\cite{kassal2008polynomial}, this source of error could lead in general to the same $\widetilde O(\exp(\eta d))$ scaling encountered with real-space simulation methods, although this issue might be mitigated in practice by choosing appropriate basis functions. Assumptions about properties of a basis could also be unreliable for more general applications such as optimization.

\paragraph{Contributions.}
In this paper, we propose efficient quantum simulation algorithms in first quantization for multi-particle real-space Schr{\"o}dinger equations.

Our primary consideration is to control the error in a discrete approximation of the continuum solution of a Schr{\"o}dinger equation. We perform spatial discretization using the Fourier spectral method. Since the error of this method decreases exponentially with the number of basis functions~\cite{boyd2001chebyshev}, it provides a high-precision approximation, as characterized in \lem{pseudospectral-estimate-lemma}.
The Fourier spectral method yields a discretized Hamiltonian of the form $H = A+B$, where $A$ is the (truncated) kinetic term and $B$ is the (discretized) potential term.
We explore three different techniques to simulate this discretized Hamiltonian.

First, we develop and analyze a \emph{$k$th-order product formula method} for simulating the Hamiltonian $H = A+B$. This method uses the fact that the evolution operators $e^{-iAt}$ and $e^{-iBt}$ can be efficiently implemented, since the potential operator $B$ is diagonal in the computational basis, while the quantum Fourier transform diagonalizes the kinetic operator $A$. Such an approach was considered in early work of Wiesner \cite{wiesner1996simulations} and Zalka \cite{zalka1998efficient}, although they did not rigorously bound the complexity. The $k$th-order product formula method uses $O(5^{2k}(\|H\|T)^{1+1/2k}/\epsilon^{1/2k})$ exponentials to simulate $H$ for time $T$ with error at most $\epsilon$~\cite{berry2007efficient}. Furthermore, its empirical performance can be better in practice than other Hamiltonian simulation methods for modestly sized classically hard instances of particular models such as spin systems~\cite{childs2018toward}. Combining the error analysis of the Fourier spectral method and the $k$th-order product formula method, we obtain the following result.

	\begin{theorem}[Informal version of \thm{time-independent-many-Schrodinger-simulation-SO}]\label{thm:time-independent-simulation-SO}
		Consider an instance of the Schr{\"o}dinger equation \eqn{Schrodinger-indep} for $\eta$ particles in $d$ dimensions, with a time-independent potential $f(\vect{x})$ satisfying $\|f(\vect{x})\|_{L^{\infty}} \le \|f\|_{\max}$. Hamiltonian simulation with the $k$th-order product formula method can produce an approximated wave function at time $T$ on a set of grid nodes, with $\ell_2$ error at most $\epsilon$, with asymptotic gate complexity
        \begin{align}
			\widetilde O\Bigl( 5^{2k}\eta d(\eta d+\|f\|_{\max})^{1+1/2k}T^{1+1/2k}/\epsilon^{1/2k}\cdot\poly\bigl(\log(\eta dTg'/\epsilon)\bigr)\Bigr),
		\end{align}
        where $g'$ defined in \eqn{gprime} upper bounds the high-order derivatives of the wave function.
        \end{theorem}

    Second, we combine the Fourier spectral method with the \emph{truncated Taylor series} approach to Hamiltonian simulation~\cite{berry2015simulating} to develop high-precision real-space simulations. We provide a concrete complexity analysis of truncating the Fourier series and performing Hamiltonian simulation in position space, and achieve the following result for any bounded potential.

	\begin{theorem}[Informal version of \thm{time-independent-many-Schrodinger-simulation}]\label{thm:time-independent-simulation}
		Consider an instance of the Schr{\"o}dinger equation \eqn{Schrodinger-indep} for $\eta$ particles in $d$ dimensions, with a time-independent potential $f(\vect{x})$ satisfying $\|f(\vect{x})\|_{L^{\infty}} \le \|f\|_{\max}$. Hamiltonian simulation with a truncated Taylor series can produce an approximated wave function at time $T$ on a set of grid nodes, with $\ell_2$ error at most $\epsilon$, with asymptotic gate complexity
		\begin{align}
			\eta d(\eta d+\|f\|_{\max})T\poly\bigl(\log(\eta dTg'/\epsilon)\bigr),
		\end{align}
        where $g'$ defined in \eqn{gprime} upper bounds the high-order derivatives of the wave function.
	\end{theorem}
\noindent
This result resolves the exponential scaling problem of Kivlichan et al.~\cite{kivlichan2017bounding}, exponentially improves the dependence on $\epsilon$, and polynomially improves the dependence on $\eta$, $d$, and $T$.

Third, we also consider the case in which the potential function $f(\vect{x},t)$ has an explicit time dependence. We assume $f(\vect{x},t)$ is bounded for any $t \ge 0$ and Lipschitz continuous in terms of $t$. We apply the Fourier spectral method to the time-dependent Hamiltonian to obtain a discretized Hamiltonian of the form $H = A + B(t)$, where $A$ represents the kinetic operator $-\frac{1}{2}\nabla^2$ and $B(t)$ captures the potential function $f(\vect{x},t)$. The matrix $A$ is an approximation of an unbounded operator, so its spectral norm is usually much larger than that of $B(t)$. Since $A$ is diagonalized in Fourier basis, we can give an efficient implementation of $e^{-iAt}$ for any $t$. Under such conditions, it is natural to apply \emph{interaction picture simulation} \cite{low2018hamiltonian}.
Combining that approach with the rescaled Dyson-series algorithm \cite{berry2020time}, we obtain the following result.

\begin{theorem}[Informal version of \thm{many-rescaled-interaction-complexity}]\label{thm:time-dependent-simulation}
	Consider an instance of the Schr{\"o}dinger equation \eqn{Schrodinger-indep} for $\eta$ particles in $d$ dimensions, with a time-dependent potential $f(\vect{x},t)$ that is bounded for any fixed $t \ge 0$ and is $L$-Lipschitz continuous in $t$. Hamiltonian simulation with a rescaled Dyson series and interaction picture can produce an approximated wave function at time $T$ on a set of grid nodes, with $\ell_2$ error at most $\epsilon$, with asymptotic gate complexity
	\begin{align}
		\eta d\|f\|_{\max,1} \poly\bigl(\log(L\|f\|_{\max,1}g'/\epsilon)\bigr),
	\end{align}
	where $\|f\|_{\max,1}:=\int_{0}^T\|f(t)\|_{\max}\,\d t$ measures the integrated strength of the potential $f$, and $g'$ defined in \eqn{gprime} upper bounds the high-order derivatives of the wave function.
\end{theorem}

We can compare these results as follows. \thm{time-independent-simulation-SO}, based on product formulas, gives a real-space quantum simulation algorithm with significant improvements in $\eta, d, T, 1/\epsilon$ compared to the previous state-of-the-art result~\cite{kivlichan2017bounding}. This method may be advantageous since product formulas are conceptually simple and often perform well in practice.
\thm{time-independent-simulation}, based on a truncated Taylor series, achieves high-precision real-space quantum simulation with poly-logarithmic scaling in $1/\epsilon$.
Given sufficient information about the Hamiltonian, \thm{time-dependent-simulation} uses the interaction picture method to address high-precision real-space quantum simulation with time-dependent potentials, and further improves the dependence on $d$ to give our best asymptotic bound. Moreover, by exploiting the techniques from the rescaled Dyson-series algorithm \cite{berry2020time}, we achieve an $L^1$-norm scaling in terms of $f$, which is advantageous when the maximum value of $f$ changes significantly during the simulation time. We emphasize that the results of our last two methods scale as $\poly(\log(1/\epsilon))$, while previous first-quantized simulations scale as $\poly(1/\epsilon)$.

For a black-box time-independent potential $f$ with an upper bound $M$ on each pairwise interaction, we have $\|f(t)\|_{\max}\leq M\binom{\eta}{2}=M\eta(\eta-1)/2$ and $\|f\|_{\max,1}\leq M\eta(\eta-1)T/2$. Therefore, by \thm{time-dependent-simulation}, the quantum gate complexity of simulating the Schr{\"o}dinger equation is
\begin{align}
M\eta^{3}dT \poly\bigl(\log(ML\eta Tg'/\epsilon)\bigr).
\end{align}
However, in concrete applications, we need to implement the black box for the potential. We explore this by considering the real-space dynamics of $\eta$ charged particles in $d$ dimensions under a generalized Coulomb potential. This scenario raises two issues. First, the Coulomb potential is unbounded for arbitrarily close particles. This can be addressed by considering a modified potential with an approximation parameter $\Delta$ such that $M\leq O(1/\Delta)$, as defined formally in Eq.~\eqn{modified-Coulomb}. Second, we must consider the computational cost of evaluating the pairwise interactions. The generalized Coulomb potential can be computed either by directly summing pairwise interactions, with cost $O(\eta^2)$; or by more advanced numerical techniques for $\eta$-body problems such as the multipole-based Barnes-Hut simulation algorithm~\cite{barnes1986hierarchical}, with cost linear in $\eta$ when the dimension $d$ is a constant. Taking these issues into account, we obtain the following result.

\begin{corollary}[Informal version of \cor{interaction-coulomb}]
\label{cor:interaction-coulomb-informal}
Consider an instance of \thm{time-dependent-simulation} where $f(\vect{x})$ is the modified Coulomb potential \eqn{modified-Coulomb} with $q_i=q$ for all $1\leq i\leq\eta$. Then the $\eta$-particle Hamiltonian can be simulated for time $T$ with accuracy $\epsilon$, with either of the following costs:
\begin{enumerate}
\item $\eta^{3}(d+\eta)T\poly(\log(\eta dTg'/(\Delta\epsilon)))/\Delta$ one- and two-qubit gates, or
\item $\eta^{3}(4d)^{d/2}T\poly(\log(\eta dTg'/(\Delta\epsilon)))/\Delta$ one- and two-qubit gates and QRAM operations, if $\Delta$ is chosen small enough that the intrinsic simulation error due to the difference between the actual Coulomb potential and the modified Coulomb potential is $O(\epsilon)$.
\end{enumerate}
\end{corollary}

While the details of how to choose $\Delta$ for a given application are outside the scope of this paper, we expect the modified Coulomb potential to give a good approximation of the actual Coulomb potential provided $\Delta$ is small compared to the minimum distances between particles, as discussed further in \sec{multi-particles}. Since this intrinsic simulation error must be small for the modified Coulomb potential to be relevant, the extra condition in the second simulation of \cor{interaction-coulomb-informal} should be satisfied in practice.

In \tab{alg-compare}, we compare the gate complexities of our methods with previous first-quantized methods for simulating the real-space dynamics of $d$-dimensional $\eta$-electron Schr{\"o}dinger equations with the potential $f(\vect{x})$ satisfying $\|f(\vect{x},t)\|_{L^{\infty}} \le \|f\|_{\max}$.
We let $\epsilon$ denote the real-space error in $\ell_2$ norm, including contributions from both the spatial discretization and the time discretization.
The quantity $g'$ defined in \eqn{gprime} upper bounds the high-order derivatives of the wave function.
The evolution time is denoted by $T$.

Previous work gave the complexities of discretized Hamiltonian simulations as a function of the grid spacing $h$ or the number of grid points $N$~\cite{kivlichan2017bounding,su2021fault}. However, such dependence can contribute additional polynomial factors of $\eta,T,\epsilon,g'$ to the complexity. For instance, Ref.~\cite{kivlichan2017bounding} considers a $d$-dimensional real-space simulation discretized on a grid by the central finite difference method. The complexity of this simulation is $\widetilde O((\eta d/h^2+\|f\|_{\max})T)$~\cite[Theorem 3]{kivlichan2017bounding}. Since $h = O(\epsilon/\eta d(g'+\eta^2T))$~\cite[Theorem 4 and Corollary 5]{kivlichan2017bounding}, the complexity of the real-space simulation is $\widetilde O\bigl(\eta^7d^3T^3(g')^2/\epsilon^2)$, as shown in \tab{alg-compare}.

Compared to Kivlichan et al.~\cite{kivlichan2017bounding}, we exponentially improve the dependence on $\epsilon$ and $g'$ from $\poly(g'/\epsilon)$ to $\poly(\log(g'/\epsilon))$.
In addition, we polynomially improve the dependence on $\eta$, $d$, $T$ to be linear, avoiding an additional polynomial dependence of these parameters when discretizing the space as in~\cite{kivlichan2017bounding}.

    \begin{table}
    \begin{center}
\resizebox{\columnwidth}{!}{
    \setlength{\tabcolsep}{4pt}
    \begin{tabular}{|c|c|c|c|c|c|}
      \hline
      \textbf{Reference} & \textbf{Representation} & \textbf{Algorithm} & \textbf{Complexity} \\
      \hline
      Kivlichan et al.~\cite{kivlichan2017bounding} & Finite Difference & Taylor & $\eta d(\eta^6 d^2T^2(g')^2/\epsilon^2+\|f\|_{\max})T \poly\bigl(\log(\eta dTg'/\epsilon)\bigr)$ \\
      \thm{time-independent-simulation} & Spectral Method & Taylor & $\eta d(\eta d+\|f\|_{\max})T\poly\bigl(\log(\eta dTg'/\epsilon)\bigr)$ \\
      \thm{time-dependent-simulation} & Spectral Method & Interaction Picture & $\eta d\|f\|_{\max}T\poly\bigl(\log(\eta dTg'/\epsilon)\bigr)$ \\
      \hline
    \end{tabular}
    }
    \end{center}
    \caption{Gate complexity comparison of for simulating the real-space dynamics of a $d$-dimensional $\eta$-particle Schr{\"o}dinger equation with the potential $f(\vect{x})$ satisfying $\|f(\vect{x},t)\|_{L^{\infty}} \le \|f\|_{\max}$. Here $T$ is the evolution time, $\epsilon$ is the $\ell_2$ real-space error, and $g'$ denoted in \eqn{gprime} upper bounds the high-order derivatives of the wave function.
    \label{tab:alg-compare}
    }
    \end{table}

Most previous work does not explicitly consider the quantum gate complexity as a function of $d$ for simulations of general $d$-dimensional $\eta$-particle Schr\"{o}dinger equations. An exception is \cite{kivlichan2017bounding}, whose cost scales as $O(d^4)$. In contrast, our \cor{interaction-coulomb} (based on the interaction picture algorithm)
can achieve quantum gate complexity $(\eta^3d+\eta^4)T \poly\bigl(\log(\eta dT/(\Delta\delta\epsilon))\bigr)/\Delta$ when $d$ is large, polynomially improving the dependence on $d$ to $\tilde{O}(d)$.

\paragraph{Applications.}

First, we consider the application of our algorithms to quantum chemistry.
As suggested by Kassal et al.~\cite{kassal2008polynomial}, direct simulation of the full quantum chemical dynamics may be more accurate and efficient than using the Born-Oppenheimer approximation, making this a potentially promising application of real-space simulation. We consider the exact molecular Schr{\"o}dinger equation under the interaction of time-independent electron-electron, electron-nucleus, and nucleus-nucleus Coulomb potentials. We then apply the Fourier spectral method and interaction picture simulation to develop an efficient real-space simulation. Using \thm{time-dependent-simulation}, we derive the following gate complexity.

\begin{corollary}[Informal version of \cor{molecular-dynamics-simulation}]\label{cor:molecular-simulation}
    Consider an instance of the Schr{\"o}dinger equation for $\eta_e$ electrons and $\eta_n$ nuclei in three spatial dimensions under the Coulomb interaction \eqn{modified-Coulomb}, where the nucleus has mass number $M$ and atomic number $Z$. These molecular dynamics can be simulated on a quantum computer within error $\epsilon$ with
	\begin{align}
        (\eta_e+M\eta_n)^3 TZ^2/(M\Delta)\cdot\poly(\log((\eta_e+M\eta_n)Tg'/(\Delta\epsilon)))
	\end{align}
	one- and two-qubit gates, along with the same number (up to poly-logarithmic factors) of QRAM operations.
\end{corollary}
\noindent
Compared to the best previous result for real-space quantum simulation of chemical dynamics~\cite{kassal2008polynomial}, the above result matches the dependence of the query and gate complexity on the particle numbers $\eta_e$ and $\eta_n$, gives explicit dependence of $T$, and achieves $\poly(\log(1/\epsilon))$ dependence on $\epsilon$.

Second, we apply our interaction picture algorithm with $L^1$-norm scaling, developed in \sec{time-dep}, to the uniform electron gas model (also known as jellium), which is a simple yet powerful model in solid-state physics. Several authors have considered quantum algorithms for simulating jellium~\cite{babbush2017low,mcardle2021exploiting}. However, to the best of our knowledge, these works have not established an asymptotic bound on the simulation complexity that takes discretization error into account. Using \thm{time-dependent-simulation}, we bound the simulation cost as follows.
\begin{corollary}[Informal version of \cor{jellium-simulation}]\label{cor:informal-jellium}
The $3$-dimensional uniform electron gas model with $\eta$ electrons can be simulated for time $T$ on a quantum computer within error $\epsilon$ with
    \begin{align}\label{eqn:informal-interaction-coulomb}
        \eta^{3}T\poly(\log(\eta Tg'/(\Delta\epsilon)))/\Delta,
    \end{align}
    one- and two-qubit gates, along with the same number (up to poly-logarithmic factors) of QRAM operations.
\end{corollary}

Third, we consider possible applications of our real-space dynamics simulation algorithms in the context of optimization. Recent work~\cite{zhang2020quantum} demonstrated that for a saddle point of a high dimensional nonconvex function, one can detect its nearby negative curvature structure by simulating the Sch{\"o}dinger dynamics of a Gaussian wavepacket centered at this point. Since saddle points are ubiquitous in the landscape of nonconvex functions (see e.g.~\cite{dauphin2014identifying,fyodorov2007replica}), escaping from saddle points is one of the major difficulties in nonconvex optimization. By exploiting our interaction picture algorithm with $L^1$-norm scaling (\prop{rescaled-interaction-complexity}), we show that we can escape from saddle points and further find a local minimum of the objective function with the following cost.

\begin{corollary}[Informal version of \cor{esc-saddle-query}]
For a $d$-dimensional twice-differentiable function $f$ that is $\ell$-smooth and $\rho$-Hessian Lipschitz, and for any $\epsilon>0$, there exists a quantum algorithm that outputs an $\epsilon$-approximate local minimum with probability at least $2/3$ using $\tilde{O}\big(\frac{f(\vect{x}_0)-f^{*}
}{\epsilon^{1.75}} \log d\big)$ queries to the evaluation oracle $U_f$, where $\vect{x}_0$ is an initial point and $f^*$ is the global minimum of $f$.
\end{corollary}
\noindent
Compared to \cite{zhang2020quantum}, which uses $\tilde{O}(\log^2d/\epsilon^{1.75})$ queries to find a local minimum, our algorithm achieves a quadratic speedup in terms of $\log d$.

\paragraph{Organization.}
The rest of the paper is organized as follows. \sec{time-indep} introduces the Fourier spectral method and develops simulations of time-independent multi-particle Schr{\"o}dinger equations based on product formulas and the truncated Taylor series method. \sec{time-dep} generalizes high-precision real-space simulation to time-dependent multi-particle Schr{\"o}dinger equations by utilizing the interaction picture technique. \sec{applications} discusses several applications of our results, including quantum chemistry, the uniform electron gas, and optimization. We conclude and discuss open questions in \sec{conclusions}. \append{notation} introduces some notation used throughout the paper, and \append{spectral} establishes an error bound for the Fourier spectral method.


\section{Simulating Schr{\"o}dinger equations in real space}\label{sec:time-indep}
	\subsection{Fourier spectral method}
	\label{sec:fourier-spectral}

In this section, we develop an approach to simulating the Schr{\"o}dinger equation in real space that combines the Fourier spectral method with Hamiltonian simulation. The Fourier spectral method (also known as the Fourier pseudospectral method) provides a global approximation to the exact solution of a partial differential equation with periodic boundary conditions. This approch can be contrasted with local approximations---such as the finite difference method---that approximate the solution on a set of grid points. In general, the Fourier spectral method approximates the solution by a linear combination of Fourier basis functions with undetermined time-dependent coefficients. By interpolating the partial differential equations at uniformly spaced nodes, we obtain a system of ordinary differential equations that can be solved numerically~\cite{boyd2001chebyshev,tang2006spectral}. Applying this approach to the Schr{\"o}dinger equation, we obtain a discretized Hamiltonian system that can be handled by standard Hamiltonian simulation algorithms.

At first glance, the Fourier spectral method looks similar to plane-wave methods widely used in first-quantized quantum simulations. Although these two approaches both employ Fourier basis functions, their primary difference is that they approximate the infinite-dimensional functional space in different finite-dimensional subspaces, and in particular, result in different discretized Hamiltonian systems. To illustrate this difference, let $\Phi(x,t)$ and $\widetilde \Phi(x,t)$ denote the exact and approximated solutions, respectively, of a one-dimensional Schr{\"o}dinger equation, and let
	\begin{align}\label{eqn:Schrodinger-residue}
	  R_n(x,t) = \Big[i\frac{\partial}{\partial t}+\frac{1}{2}\nabla^{2}-f(x)\Big]\widetilde \Phi(x,t)
	\end{align}
denote the \emph{residue}, where $n+1$ is the number of the basis functions. The residue quantifies the extent to which the approximated solution fails to satisfy the Schr{\"o}dinger equation \eqn{Schrodinger-indep}. In general, $R_n$ cannot be zero as a function of $t$ unless the exact solution $\Phi$ is a finite combination of basis functions. Instead, we seek a reasonable choice of $\widetilde \Phi$ such that the projection of $R_n$ onto some finite-dimensional subspace vanishes. As we describe below in more detail, the Fourier spectral approach requires the residue to vanish at the set of interpolation nodes, while Galerkin plane wave methods instead guarantee that its integrations with Fourier test functions are zero.

In the Galerkin approach~\cite{brenner2008mathematical,szabo1991finite}, the residue is orthogonal to a subspace of $n+1$ chosen test functions, denoted $\{\phi_j\}_{j=0}^n$. In other words, we require that
    \begin{align}\label{eqn:Gakerlin-integration}
	  \langle \phi_j | R_n \rangle = \langle \phi_j | H | \widetilde \Phi \rangle = 0, \qquad j \in \rangez{n+1}=\{0,1,\ldots,n\},
	\end{align}
where angle brackets denote the inner product over the spatial domain, and $H$ is the Galerkin discretized Hamiltonian.
For the Schr{\"o}dinger equation \eqn{Schrodinger-indep}, we have $H=T+V$ where the matrix elements of the discretized kinetic and potential terms are given by \cite{su2021fault}
    \begin{align}\label{eqn:Gakerlin-element}
	  T_{pq} &= \int \d r \, \phi^{\ast}_p(x)\biggl(-\frac{\nabla^2}{2}\biggr)\phi_q(x), \\
      V_{pq} &= \int \d r \, \phi^{\ast}_p(x)f(x)\phi_q(x).
	\end{align}
(See Appendix B of \cite{su2021fault} for Galerkin representations of molecular Hamiltonians.)

Equation \eqn{Gakerlin-integration} is a system of $n+1$ ordinary differential equations with time-dependent coefficients. For first-quantized plane-wave methods, the basis functions used in constructing $\widetilde \Phi$ as well as the test functions $\{\phi_j\}$ are all chosen as Fourier basis functions. While this discretization is commonly used in first-quantized quantum simulations~\cite{toloui2013quantum,babbush2017exponentially,babbush2019quantum,su2021fault}, previous studies neglect the discretization error. If the Schr{\"o}dinger operator includes an unbounded potential or is highly oscillatory, such a Galerkin representation may not provide a reasonable approximation.

On the other hand, in the spectral approach~\cite{boyd2001chebyshev,tang2006spectral}, we choose the test functions $\{\phi_j\}$ in \eqn{Gakerlin-integration} to be delta functions on the uniform interpolation nodes $\{\chi_j\}_{j=0}^n$. Then we have
    \begin{align}\label{eqn:Gakerlin-interpolation}
	  \langle \delta_j | R_n \rangle = R_n(\chi_j,t) = 0, \qquad k \in \rangez{n+1},
	\end{align}
which is equivalent to
    \begin{align}\label{eqn:interpolation}
	  \Big[i\frac{\partial}{\partial t}+\frac{1}{2}\nabla^{2}-f(\vect{x})\Big]\widetilde \Phi(\chi_j,t) = 0, \qquad k \in \rangez{n+1}.
	\end{align}
This choice again defines a system of $n+1$ ordinary differential equations with time-dependent coefficients. The spectral approach can provide a straightforward approximation of real-space quantum dynamics by determining the number of interpolation points $n+1$ explicitly as a function of the allowed discretization error, the particle number, and the norms of high-order derivatives of the wave function. In contrast, previous simulations based on the Galerkin approach~\cite{toloui2013quantum,babbush2017exponentially,babbush2019quantum,su2021fault} did not explicitly take the real-space discretization error into account, and instead merely determined the complexity in terms of the number of basis functions used.

We now introduce our Fourier spectral approach for time-dependent Schr{\"o}dinger equations of the general form
	\begin{align}\label{eqn:Schrodinger-dep}
i\frac{\partial}{\partial t}\Phi(\vect{x},t)=\Big[-\frac{1}{2}\nabla^{2}+f(\vect{x},t)\Big]\Phi(\vect{x},t)
\end{align}
where $\vect{x}\in\R^{d}$ represents the position of the quantum particle, $t\in\R$ represents time, $\nabla=\bigl(\frac{\partial}{\partial x_{i}}\bigr)\big|_{i=1}^{d}$ is the gradient, and $f\colon\R^{d}\times\R\to\R$ is the potential function. This generalizes \eqn{Schrodinger-indep} to the case of time-dependent potentials. We assume access to the potential through a unitary oracle $U_{f}$ such that for any $\x\in\R^{d}$, $t\in\R$, and $z\in\R$,
	\begin{align}\label{eqn:evaluation-dep}
		U_{f}|\x\>|t\>|z\>=|\x\>|t\>|f(\x)+z\>.
	\end{align}

For concreteness, we consider $\vect{x}\in\Omega:=[0,1]^d$ and assume periodic boundary conditions for $\Phi(\vect{x},0)$, i.e.,
	\begin{align}\label{eqn:periodic-condition}
		\frac{\partial^{(p)}}{\partial x_j^{(p)}}\Phi(x_1, \ldots, x_{j-1}, 0, x_{j+1}, \ldots, x_d,0) = \frac{\partial^{(p)}}{\partial x_j^{(p)}}\Phi(x_1, \ldots, x_{j-1}, 1, x_{j+1}, \ldots, x_d,0)
	\end{align}
holds for all $p \in \N, ~j \in \range{d}$, where $\frac{\partial^{(p)}}{\partial x_j^{(p)}}$ is the $p$th-order partial derivative with respect to the $j$th coordinate of $\vect{x} = [x_1, \ldots, x_d]^T$.

 We first apply the Fourier spectral method for the spatial discretization. We approximate the solution $\Phi(\vect{x},t)$ by a Fourier series of the form
	\begin{align}\label{eqn:basis-space-expand}
		\widetilde \Phi(\vect{x},t) = \sum_{\|\vect{k}\|_{\infty}\le n}c_{\vect{k}}(t)\phi_{\vect{k}}(\vect{x})
	\end{align}
for some even number $n \in \N$, where $\vect{k}=(k_1,\ldots,k_d)$ with $k_j\in \rangez{n+1}$, $c_{\vect{k}}(t)\in \C$, and
	\begin{align}\label{eqn:basis_tensor}
		\phi_{\vect{k}}(\vect{x}) = \prod_{j=1}^d\phi_{k_j}(x_j)
	\end{align}
with
	\begin{align}\label{eqn:basis_function}
		\phi_{k}(x) := e^{2\pi i(k-n/2) x}
	\end{align}
for $k \in \rangez{n+1}$ and $x \in [0,1]$.

	Plugging \eqn{basis-space-expand} into \eqn{Schrodinger-dep}, we obtain an approximated PDE system
	\begin{align}
		i\frac{\partial}{\partial t}\widetilde \Phi(\vect{x},t)=\Big[-\frac{1}{2}\nabla^{2}+f(\vect{x},t)\Big]\widetilde \Phi(\vect{x},t)
	\end{align}
    with the initial condition
    \begin{align}
		\widetilde \Phi(\vect{x},0) = \Phi(\vect{x},0).
	\end{align}

In terms of the basis functions, this gives
	\begin{align}\label{eqn:approximated_PDE}
		i\sum_{\|\vect{k}\|_{\infty}\le n}\frac{\d}{\d{t}}c_{\vect{k}}(t)\phi_{\vect{k}}(\vect{x}) = \sum_{\|\vect{k}\|_{\infty}\le n}c_{\vect{k}}(t)\Big[-\frac{1}{2}\sum_{\|\vect{r}\|_{\infty}\le n}[{\vect{L}}_{n,d}]_{\vect{kr}}\phi_{\vect{r}}(\vect{x})+f(\vect{x},t)\phi_{\vect{k}}(\vect{x})\Big]
	\end{align}
	where ${\vect{L}}_{n,d}$ is the multi-dimensional Laplacian matrix
	\begin{align}
		{\vect{L}}_{n,d} := \bigoplus_{j=1}^d D^2_n = D^2_n\otimes I^{\otimes d-1}+I\otimes D^2_n\otimes I^{\otimes d-2}+\cdots+I^{\otimes d-1}\otimes D^2_n
	\end{align}
and $D_n$ is a differential matrix for the Fourier basis functions \eqn{basis_function}, the $(n+1)$-dimensional diagonal matrix with entries
	\begin{align}\label{eqn:Fourier_Dn}
		[D_n]_{kk}=2\pi i(k-n/2)
	\end{align}
	for $k \in \rangez{n+1}$.

	To produce a system of ordinary differential equations, we introduce the \emph{uniform interpolation nodes} $\{\vect{\chi}_{\vect{l}}=(\chi_{l_1}, \ldots, \chi_{l_d})\}_{\|\vect{l}\|_{\infty}\le n}$ with $l_j\in \rangez{n+1}$, where
	\begin{align}\label{eqn:interpolation_nodes}
		\chi_{l_j} = \frac{l_j}{n+1}.
	\end{align}
	Considering \eqn{approximated_PDE} at the uniform interpolation nodes \eqn{interpolation_nodes}, we obtain an $(n+1)^d$-dimensional approximated ODE system
	\begin{align}\label{eq:odesystem}
		&i\sum_{\|\vect{k}\|_{\infty}\le n}\frac{\d}{\d{t}}c_{\vect{k}}(t)\phi_{\vect{k}}(\vect{\chi}_{\vect{l}})|l_1\rangle\ldots|l_d\rangle\nonumber\\
&\qquad = \sum_{\|\vect{k}\|_{\infty}\le n}c_{\vect{k}}(t)\Big[-\frac{1}{2}\sum_{\|\vect{r}\|_{\infty}\le n}[{\vect{L}}_n]_{\vect{kr}}\phi_{\vect{r}}(\vect{\chi}_{\vect{l}})+f(\vect{\chi}_{\vect{l}},t)\phi_{\vect{k}}(\vect{\chi}_{\vect{l}})\Big]|l_1\rangle\ldots|l_d\rangle,
	\end{align}
	where $l_j\in\rangez{n+1},~ j\in\range{d}$.

    The well-known \emph{quantum Fourier transform} (QFT) maps the $(n+1)$-dimensional quantum state $v=(v_0,v_1,\ldots,v_n)\in\C^{n+1}$ to the state $\hat{v}=(\hat{v}_0,\hat{v}_1,\ldots,\hat{v}_n)\in\C^{n+1}$ with
    \begin{equation}\label{eqn:qft_rule}
    \hat{v}_l = \frac{1}{\sqrt{n+1}}\sum_{k=0}^{n}\exp\Bigl(\frac{2\pi ikl}{n+1}\Bigr)v_k,\quad l\in\rangez{n+1}.
    \end{equation}
    In other words, the QFT is the unitary transform
    \begin{equation}\label{eqn:qft_matrix}
    F_n:=\frac{1}{\sqrt{n+1}}\sum_{k,l=0}^n \exp\Bigl(\frac{2\pi ikl}{n+1}\Bigr)|l\rangle\langle k|.
    \end{equation}
	The closely related \emph{quantum shifted Fourier transform} (QSFT) maps the $(n+1)$-dimensional quantum state $v \in \C^{n+1}$ to the state $\hat v \in \C^{n+1}$ with
    \begin{equation}\label{eqn:qsft_rule}
    \hat{v}_l = \frac{1}{\sqrt{n+1}}\sum_{k=0}^{n}\exp\Bigl(\frac{2\pi i(k-n/2)l}{n+1}\Bigr)v_k,\quad l\in\rangez{n+1}.
    \end{equation}
    In other words, the QSFT is the unitary transform
	\begin{align}\label{eqn:qsft_matrix}
		F^s_n := \frac{1}{\sqrt{n+1}}\sum_{k,l=0}^n \exp\Bigl(\frac{2\pi i(k-n/2)l}{n+1}\Bigr)|l\rangle\langle k|.
	\end{align}
    Notice that the QSFT can be written as the product
    \begin{equation}
    F^s_n=S_nF_n,
    \end{equation}
    of the QFT defined above and the diagonal matrix
    \begin{equation}
    S_n=\sum_{l=0}^n \exp\Bigl(-\frac{\pi i n l}{n+1}\Bigr)|l\rangle\langle l|.
    \end{equation}
    The QSFT can be performed with gate complexity $O(\log n \log\log n)$ \cite[Lemma 5]{childs2020high}. Using \eqn{basis-space-expand} in the one-dimensional case with $v_k = c_k(t)$, the QSFT maps the state $v$ to $\hat v = F^s_n v$ satisfying
    \begin{equation}
    \hat{v}_l = \frac{1}{\sqrt{n+1}}\sum_{k=0}^{n}c_k(t)\phi_{k}(\chi_{l}) = \frac{1}{\sqrt{n+1}}\widetilde \Phi(\chi_{l},t),\quad l\in\rangez{n+1}.
    \end{equation}
    In other words, the QSFT maps the coefficient vector $v = \sum_{k=0}^{n} c_k(t)|k\rangle$ to approximate interpolated solution $\hat v = \frac{1}{\sqrt{n+1}}\sum_{l=0}^{n} \widetilde \Phi(\chi_{l},t)|l\rangle$.
    We use the QSFT (instead of the ordinary QFT) to align with the phase convention specified in \eqn{basis-space-expand}.

We also define the multi-dimensional QSFT as
	\begin{align}
		\vect{F}^s_{n,d} := \bigotimes_{j=1}^d F^s_n.
	\end{align}

	Letting
	\begin{align}
		c(t) := \sum_{\|\vect{k}\|_{\infty}\le n}c_{\vect{k}}(t)|k_1\rangle\ldots|k_d\rangle, \qquad |c(t)\rangle := \frac{c(t)}{\|c(t)\|},
	\end{align}
and
	\begin{align}
		\vect{V}_{n,d}(t) := \sum_{\|\vect{l}\|_{\infty}\le n}f(\vect{\chi}_{\vect{l}},t)|l_1\rangle\ldots|l_d\rangle\langle l_1|\ldots\langle l_d|,
	\end{align}
	the ODE system \eq{odesystem} can be rewritten as
	\begin{align}
		i\vect{F}^s_{n,d}\frac{\d}{\d{t}}|c(t)\rangle = \vect{F}^s_{n,d}{\vect{L}}_{n,d}|c(t)\rangle+\vect{V}_{n,d}(t)\vect{F}^s_{n,d}|c(t)\rangle.
	\end{align}
Equivalently,
	\begin{align}\label{eqn:semi_discrete_ODE_frequency}
		i\frac{\d}{\d{t}}|c(t)\rangle = {\vect{H}}_{n,d}(t)|c(t)\rangle = [\vect{L}_{n,d}+(\vect{F}^s_{n,d})^{-1}\vect{V}_{n,d}(t)\vect{F}^s_{n,d}]|c(t)\rangle
	\end{align}
	which is a Hamiltonian system in the momentum space,
    with the Hamiltonian
	\begin{align}\label{eqn:spectral_Hamiltonian_frequency}
		{\vect{H}}_{n,d}(t) := \vect{L}_{n,d}+(\vect{F}^s_{n,d})^{-1}\vect{V}_{n,d}(t)\vect{F}^s_{n,d}.
	\end{align}

	Alternatively, \eqn{semi_discrete_ODE_frequency} can be expressed as
	\begin{align}\label{eqn:semi_discrete_ODE_transform}
		i\frac{\d}{\d{t}}[\vect{F}^s_{n,d}|c(t)\rangle] = [\vect{F}^s_{n,d}{\vect{L}}_{n,d}(\vect{F}^s_{n,d})^{-1}+\vect{V}_{n,d}(t)][\vect{F}^s_{n,d}|c(t)\rangle].
	\end{align}
	Using \eqn{basis-space-expand}, we
	write
	\begin{align}\label{eqn:position-state-approximate}
		\widetilde \Phi(t) = \sum_{\|\vect{l}\|_{\infty}\le n}\widetilde \Phi(\vect{\chi}_{\vect{l}},t)|l_1\rangle\ldots|l_d\rangle = \sum_{\|\vect{l}\|_{\infty}\le n}c_{\vect{k}}(t)\phi_{\vect{k}}(\vect{\chi}_{\vect{l}})|l_1\rangle\ldots|l_d\rangle, \qquad |\widetilde \Phi(t)\rangle := \frac{\widetilde \Phi(t)}{\|\widetilde \Phi(t)\|},
	\end{align}
	such that
	\begin{align}\label{eqn:frequency-position-transform}
		\widetilde \Phi(t) = \vect{F}^s_{n,d}c(t), \qquad |\widetilde \Phi(t)\rangle = \vect{F}^s_{n,d}|c(t)\rangle
	\end{align}
	provide approximations of the exact solution and its $\ell_2$ normalized state
	\begin{align}\label{eqn:position-state-exact}
		\Phi(t) = \sum_{\vect{l}}\Phi(\vect{\chi}_{\vect{l}},t)|l_1\rangle\ldots|l_d\rangle, \qquad |\Phi(t)\rangle := \frac{\Phi(t)}{\|\Phi(t)\|},
	\end{align}
	respectively.
	Thus we see that Eq.~\eqn{semi_discrete_ODE_frequency} is a Hamiltonian system in position space
    \begin{align}\label{eqn:semi_discrete_ODE_time}
		i\frac{\d}{\d{t}}|\widetilde \Phi(t)\rangle = \vect{\widetilde H}_{n,d}(t)|\widetilde \Phi(t)\rangle = [\vect{F}^s_{n,d}\vect{L}_{n,d}(\vect{F}^s_{n,d})^{-1}+\vect{V}_{n,d}(t)]|\widetilde \Phi(t)\rangle,
	\end{align}
with the Hamiltonian
	\begin{align}\label{eqn:spectral_Hamiltonian_time}
		{\vect{\widetilde H}}_{n,d}(t) &:= \vect{F}^s_{n,d}{\vect{L}}_{n,d}(\vect{F}^s_{n,d})^{-1}+\vect{V}_{n,d}(t) \\
		&= \vect{F}^s_{n,d}{\vect{H}}_{n,d}(t)(\vect{F}^s_{n,d})^{-1}.
		\label{eqn:spectral_Hamiltonian_interaction}
	\end{align}

    Furthermore, we assume the $\ell_2$ norm of the exact $(n+1)^d$-dimensional initial condition $\Phi(0)$ satisfies
    \begin{align}\label{eqn:initial-norm}
		\|\Phi(0)\|^2 = \sum_{\|\vect{l}\|_{\infty}\le n}|\Phi(\vect{\chi}_{\vect{l}},0)|^2 = (n+1)^d.
	\end{align}
    This is a discrete analog of the condition $\int_{\vect{x}\in\Omega}|\Phi(\vect{x},0)|^2 \, \d{\vect{x}} = 1$ on the $(n+1)^d$ uniform interpolation nodes $\{\vect{\chi}_{\vect{l}}\}$. In more detail, consider the trapezoidal rule for numerical integration~\cite{hildebrand1987introduction}. On each $d$-dimensional grid cell, we replace the integration of $\Phi(\vect{x},0)$ by the average value of $2d$ nearby interpolation points $\Phi(\vect{\chi}_{\vect{l}},0)$ times the volume $\frac{1}{(n+1)^d}$ of the $d$-dimensional grid cell. In this setting, $\|\Phi(0)\|^2/(n+1)^d$ approximates $\int_{\vect{x}\in\Omega}|\Phi(\vect{x},0)|^2 \, \d{\vect{x}}$. For convenience, we normalize the state according to \eqn{initial-norm}.\footnote{Given an arbitrary initial condition $\Phi(\vect{x},0)$ and its corresponding discretized state $\Phi(0)$, we can rescale the initial condition as $\Phi(\vect{x},0) \to \frac{(n+1)^{d/2}}{c}\Phi(\vect{x},0)$ and $\Phi(0) \to \frac{(n+1)^{d/2}}{c}\Phi(0)$ where $c := \|\Phi(0)\|$, such that the rescaled state satisfies $\|\Phi(0)\| = (n+1)^d$.}
    Because the Schr{\"o}dinger equation is unitary, this ensures that
    \begin{align}\label{eqn:conservation-norm}
		\|\Phi(t)\|^2 = \sum_{\|\vect{l}\|_{\infty}\le n}|\Phi(\vect{\chi}_{\vect{l}},t)|^2 = (n+1)^d
	\end{align}
    for all $t \in \R$.

	\subsection{Truncation number of the Fourier spectral method}

	The overall simulation error includes two contributions: the error introduced by discretizing the problem with the Fourier spectral method and the error introduced by the Hamiltonian simulation algorithm. To ensure overall error at most $\epsilon$, we choose the parameters of the Fourier spectral method to upper bound the spatial discretization error between the exact and approximated normalized states ($|\Phi(t)\rangle$ and $|\widetilde \Phi(t)\rangle$, respectively) by $\epsilon/2$, and choose the parameters of the Hamiltonian simulation algorithm to also upper bound its error by $\epsilon/2$. The latter calculation uses standard analysis to bound the error accumulated over the course of the simulation. In the following, we analyze the error of the Fourier spectral method.

    Spectral methods typically exhibit exponential convergence
    if the solution is smooth~\cite{boyd2001chebyshev}. In particular, we establish exponential convergence for approximating $\Phi(\vect{x},t)$ by $\widetilde \Phi(\vect{x},t)$.

	\begin{lemma}\label{lem:pseudospectral-estimate-lemma}
		Let $\Phi(\vect{x},t)$ and $\widetilde \Phi(\vect{x},t)$
		denote the exact and approximated solutions of \eqn{Schrodinger-indep} by the Fourier spectral method, respectively, where $\Phi(\vect{x},t)$ is analytic in $t$ and $\vect{x}$. Then for any even integer $n\ge6$, the error from the Fourier spectral method satisfies
	\begin{align}\label{eqn:pseudospectral-estimate-Schrodinger}
		\max_{\vect{x},t}|\Phi(\vect{x},t) - \widetilde \Phi(\vect{x},t)| \le \frac{2}{\pi}\frac{\max_t\|\Phi^{(n/2)}(\cdot,t)\|_{L^1}}{(n/2)^{n/2}}.
	\end{align}
	\end{lemma}

    \lem{pseudospectral-estimate-lemma} gives an estimate of the maximal error of approximating $\Phi(\vect{x},t)$ by $\widetilde \Phi(\vect{x},t)$ in space and time. We prove \lem{pseudospectral-estimate-lemma} in \append{spectral}.

Using this error estimate, we can determine a sufficient truncation number $n$ that ensures the approximated solution $\widetilde \Phi(\vect{x},t)$ is within the allowed error tolerance. For simplicity, we denote
	\begin{align}\label{eqn:gprime}
		g' := \max_t\|\Phi^{(n/2)}(\cdot,t)\|_{L^1}.
	\end{align}
	The parameter $g'$	describes the higher-order regularity of the wave function $\Phi(\x,t)$. Usually, when $\Phi(\x,t)$ is not strongly localized, it is common to assume $g'$ is bounded from above~\cite{kivlichan2017bounding,jin2021quantum}. In fact, Bourgain~\cite{bourgain1999growth} shows that the derivatives of the wave function $\Phi(\x,t)$ are bounded when the potential function $f(\x, t)$ is sufficiently smooth. However, to our best knowledge, the exact scaling of $g'$ in terms of $d$ and $n$ remains unknown. Therefore, in the present analysis, we parametrize the overall complexity by $g'$.

Using \eqn{pseudospectral-estimate-Schrodinger}, for any $t\in\R^+$ and $\vect{x}=\vect{\chi}_{\vect{l}}$ defined in \eqn{interpolation_nodes}, we have
	\begin{align}\label{eqn:entry-error}
		\Bigl|\Phi(\vect{\chi}_{\vect{l}},t)-\widetilde \Phi(\vect{\chi}_{\vect{l}},t)\Bigr| \le \frac{2}{\pi}\frac{g'}{(n/2)^{n/2}}.
	\end{align}
	Recall that the sets of all entries of $\Phi(t)$ and $\widetilde \Phi(t)$ as presented in \eqn{position-state-exact} and \eqn{position-state-approximate} are $\{\Phi(\vect{\chi}_{\vect{l}},t)\}$ and $\{\widetilde \Phi(\vect{\chi}_{\vect{l}},t)\}$ for all $\vect{\chi}_{\vect{l}}$, respectively. Each entry of $\Phi(t)-\widetilde \Phi(t)$ is bounded by the right-hand side of \eqn{entry-error}, giving
	\begin{align}
		\Bigl\|\Phi(t)-\widetilde \Phi(t)\Bigr\|_{\infty} \le \frac{2}{\pi}\frac{g'}{(n/2)^{n/2}}
	\end{align}
	for any $t\in\R^+$. With respect to the $\ell_2$ norm, using $\|\vect{v}\|_2\le\sqrt{(n+1)^d}\|\vect{v}\|_{\infty}$ for the $(n+1)^d$-dimensional vector $\vect{v}=\Phi(t)-\widetilde \Phi(t)$, we have
	\begin{align}\label{eqn:absolute-error-0}
		\Bigl\|\Phi(t)-\widetilde \Phi(t)\Bigr\|
		\le \frac{2}{\pi}\frac{g'}{(n/2)^{n/2}}(n+1)^{d/2}
	\end{align}
	for any $t\in\R^+$.

This bound implies that the error of the normalized states $|\Phi(t)\rangle$ and $|\widetilde \Phi(t)\rangle$ satisfies
    \begin{align}\label{eqn:relative-error-0}
		\Bigl\||\Phi(t)\rangle-|\widetilde \Phi(t)\rangle\Bigr\|
		\le \frac{\bigl\|\Phi(t)-\widetilde \Phi(t)\bigr\|}{\min\{\|\Phi(t)\|,\|\widetilde \Phi(t)\|\}}
        \le \frac{\delta}{\|\Phi(t)\|-\delta},
	\end{align}
    where $\delta := \max_t \bigl\|\Phi(t)-\widetilde \Phi(t)\bigr\|$. Recall that $\|\Phi(t)\| = (n+1)^{d/2}$ by \eqn{conservation-norm}. To satisfy the inequality
    \begin{align}
		 \frac{\delta}{\|\Phi(t)\|-\delta} = \frac{\delta}{(n+1)^{d/2}-\delta} \le \epsilon/2 \quad \Longleftrightarrow \quad \delta \le \frac{\epsilon/2}{1+\epsilon/2}(n+1)^{d/2},
	\end{align}
    based on \eqn{absolute-error-0}, we choose $n$ so that
	\begin{align}\label{eqn:error-0}
		\frac{2}{\pi}\frac{g'}{(n/2)^{n/2}}(n+1)^{d/2} \le \frac{\epsilon/2}{1+\epsilon/2}(n+1)^{d/2},
	\end{align}
	which is equivalent to
	\begin{align}
		(n/2)^{n/2} \ge \frac{4g'(1+\epsilon/2)}{\pi\epsilon}.
	\end{align}
    Since $\epsilon/2\le1$, and noticing the condition $n\ge6$ in \lem{pseudospectral-estimate-lemma}, it suffices to select
	\begin{align}\label{eqn:truncated-number}
		n = \max\biggl\{ 2\Bigl\lceil\frac{\log(\omega)}{\log(\log(\omega))}\Bigr\rceil, 6 \biggr\},
	\end{align}
	where
	\begin{align}
		\omega = \frac{4g'}{\pi\epsilon}.
	\end{align}
	
\subsection{Hamiltonian simulation with product formulas}

Early work of Wiesner~\cite{wiesner1996simulations} and Zalka~\cite{zalka1998efficient} used the so-called split-operator method to simulate real-space quantum dynamics on quantum computers. This method uses the truncated Fourier series to discretize the Schr{\"o}dinger equation in space and construct a discrete Hamiltonian system, with the Hamiltonian as a sum of the potential and kinetic operators. The diagonal potential operator is encoded in the position space, and the kinetic operator is diagonalized by quantum Fourier transform in the momentum space. The kinetic and potential operators are propagated independently, and these time evolutions are combined using product formulas.
Subsequently, Kassal et al.~\cite{kassal2008polynomial} applied this method to chemical dynamics. However, these previous works do not provide rigorous error analysis.
In particular, they all replace the continuous kinetic operator by the discretized one without accounting for the discretization error, as discussed in~\cite[Theorem 4]{kivlichan2017bounding}.

Having derived the Fourier spectral method with concrete real-space error analysis to obtain the discrete Hamiltonian system \eqn{semi_discrete_ODE_time}, we now describe the simulation using product formulas. Given a Hamiltonian $\vect{H}=\vect{A}+\vect{B}$, the standard $2k$th-order Suzuki product formula \cite{suzuki1991general} is defined recursively as
\begin{align}\label{eqn:Suzuki-time}
\mathscr{S}^T_2(t) &:= e^{-i\frac{t}{2}\vect{A}}\cdot e^{-it\vect{B}}\cdot e^{-i\frac{t}{2}\vect{A}}, \\
\mathscr{S}^T_{2k}(t) &:= \mathscr{S}_{2k-2}(u_kt)^2\mathscr{S}_{2k-2}((1-4u_k)t)\mathscr{S}_{2k-2}(u_kt)^2
\end{align}
where $u_k := 1/(1-4^{1/(2k-1)})$.
In our problem, the Hamiltonian ${\vect{\widetilde H}}_{n,d}(t)$ in \eqn{spectral_Hamiltonian_time} is the sum of
	\begin{align}\label{eqn:spectral_Hamiltonian_time_sum}
		\vect{A} &= \vect{F}^s_{n,d}{\vect{L}}_{n,d}(\vect{F}^s_{n,d})^{-1}, \\
		\vect{B} &= \vect{V}_{n,d}(t).
	\end{align}
Instead of directly simulating $\vect{A}$, we observe that
$e^{-it\vect{F}^s_{n,d}\vect{L}_{n,d}(\vect{F}^s_{n,d})^{-1}} = \vect{F}^s_{n,d} e^{-it\vect{L}_{n,d}}(\vect{F}^s_{n,d})^{-1}$, i.e., the evolution in the position space coincides with the Fourier transform of the evolution of $e^{-it\vect{L}_{n,d}}$ in the momentum space. In other words,
\begin{align}\label{eqn:equivalence-FLF}
i\frac{\d}{\d{t}}|\widetilde \Phi(t)\rangle = \vect{F}^s_{n,d}\vect{L}_{n,d}(\vect{F}^s_{n,d})^{-1}|\widetilde \Phi(t)\rangle\quad \Longleftrightarrow\quad
i\frac{\d}{\d{t}}|c(t)\rangle = \vect{L}_{n,d}|c(t)\rangle,
\end{align}
where $|c(t)\rangle = (\vect{F}^s_{n,d})^{-1}|\widetilde \Phi(t)\rangle$ by \eqn{frequency-position-transform}.
Therefore, the split-operator method with the $k$th-order Suzuki product formula for simulating \eqn{semi_discrete_ODE_time} can be presented recursively as
\begin{align}\label{eqn:split-operator-VLV}
\mathscr{S}^B_2(t) &= (\vect{F}^s_{n,d})^{-1}e^{-i\frac{t}{2}\vect{V}_{n,d}}\cdot\vect{F}^s_{n,d}e^{-it\vect{L}_{n,d}}\cdot(\vect{F}^s_{n,d})^{-1}e^{-i\frac{t}{2}\vect{V}_{n,d}}, \\
\label{eqn:split-operator-VLV-2k}
\mathscr{S}^B_{2k}(t) &= \mathscr{S}_{2k-2}(u_kt)^2\mathscr{S}_{2k-2}((1-4u_k)t)\mathscr{S}_{2k-2}(u_kt)^2.
\end{align}

We now give concrete upper bounds on the gate complexities of the $k$th-order split-operator method for simulating the discretized Schr{\"o}dinger equation.

  \begin{lemma}\label{lem:time-independent-Hamiltonian-simulation-SO}
		Consider an instance of time-independent Hamiltonian simulation as defined in \eqn{spectral_Hamiltonian_time}, with a time-independent potential $f(\vect{x})$ satisfying $\|f(\vect{x})\|_{L^{\infty}} \le \|f\|_{\max}$, for time $T>0$. Let $g' = \max_t\|\Phi^{(n/2)}(\cdot,t)\|_{L^1}$ as in \eqn{gprime}. There exists a quantum algorithm producing a normalized state that approximates $|\widetilde \Phi(T)\rangle$ with $\ell_2$ error at most $\epsilon/2$, with gate complexity
		\begin{align}\label{eqn:time-independent-Hamiltonian-simulation-SO}
            \widetilde O\Bigl( 5^{2k} d(d+\|f\|_{\max})^{1+1/2k}T^{1+1/2k}/\epsilon^{1/2k} \Bigr).
		\end{align}
	\end{lemma}

	\begin{proof}
	
We apply standard error bounds for product formulas~\cite{berry2007efficient}.
For the complexity analysis we simply need to include the additional cost of performing the quantum Fourier transform $\vect{F}^s_{n,d}$ and its inverse $(\vect{F}^s_{n,d})^{-1}$. The number of QFTs equals the number of exponentials, which is upper bounded by~\cite[Theorem 1]{berry2007efficient} with $m=2$, which shows that
\begin{align}\label{eqn:exponential-number}
N_{\mathrm{QFT}} = N_{\mathrm{exp}} \le 4 \cdot 5^{2k} (2\|H\|T)^{1+1/2k}/(\epsilon/2)^{1/2k}.
\end{align}
exponentials suffice to ensure that we approximate $|\widetilde \Phi(T)\rangle$ with $\ell_2$ error at most $\epsilon/2$. Using $\|\vect{L}_{n,d}\|\le\frac{dn^2}{4}$ and $\|\vect{V}_{n,d}\|\le \|f\|_{\max}$, we have
\begin{align}\label{eqn:exponential-number-reduce}
N_{\mathrm{QFT}} = N_{\mathrm{exp}} \le 4\cdot5^{2k}2^{1+1/k}
\Big[T\Big(\frac{dn^2}{4}+\|f\|_{\max}\Big)\Big]^{1+1/2k}/\epsilon^{1/2k}.
\end{align}
Since $\vect{L}^s_{n,d}$, $\vect{V}^s_{n,d}$, $\vect{F}^s_{n,d}$, and $(\vect{F}^s_{n,d})^{-1}$ can all be performed with gate complexity $d\poly(\log n)$, the cost of implementing either an exponential or an inverse quantum Fourier transform is $O(d\poly(\log n))$; the claim follows by including this factor.
\end{proof}

Combining this with an upper bound on the discretization error gives our main result on product-formula simulation.

    \begin{theorem}[$k$th-order product formula simulation of real-space dynamics]\label{thm:time-independent-Schrodinger-simulation-SO}
		Consider an instance of the Schr{\"o}dinger equation in \eqn{Schrodinger-indep} with a time-independent potential $f(\vect{x})$ satisfying $\|f(\vect{x})\|_{L^{\infty}} \le \|f\|_{\max}$ and a given $T>0$. Let $g' = \max_t\|\Phi^{(n/2)}(\cdot,t)\|_{L^1}$ as in \eqn{gprime}. There exists a quantum algorithm producing a normalized state that approximates $\Phi(\vect{x},T)$ at the nodes $\{\vect{\chi}_{\vect{l}}\}$ defined as \eqn{interpolation_nodes}, with $\ell_2$ error at most $\epsilon$, with asymptotic gate complexity
		\begin{align}\label{eqn:time-independent-Schrodinger-simulation-SO}
			\widetilde O\Bigl( 5^{2k} d(d+\|f\|_{\max})^{1+1/2k}T^{1+1/2k}/\epsilon^{1/2k} \Bigr).
		\end{align}
	\end{theorem}

	\begin{proof}
		First, by \eqn{error-0}, it suffices to take $n$ as in \eqn{truncated-number} to ensure, for any $t\in\R^+$,
		\begin{align}\label{eqn:error-1}
			\Bigl\||\Phi(t)\rangle-|\widetilde \Phi(t)\rangle\Bigr\| \le \epsilon/2.
		\end{align}
The resulting state $|\widetilde \Phi(T)\rangle$ is the solution of \eqn{spectral_Hamiltonian_time}.
        The $2k$th-order product formula \eqn{split-operator-VLV-2k} takes the gate complexity \eqn{time-independent-Hamiltonian-simulation-SO} of reaching
		\begin{align}\label{eqn:error-2}
			\Bigl\||\psi(T)\rangle-|\widetilde \Phi(T)\rangle\Bigr\| \le \epsilon/2.
		\end{align}

		Combining \eqn{error-1} with \eqn{error-2}, and taking $t=T$, we have
		\begin{align}\label{eqn:error-3}
			\Bigl\||\psi(T)\rangle-|\Phi(T)\rangle\Bigr\| \le \epsilon.
		\end{align}
		The gate complexity of producing $|\psi(T)\rangle$ is given by \eqn{time-independent-Hamiltonian-simulation-SO}, and the claimed result follows.
	\end{proof}

Comparing with the gate complexity in \cite{kivlichan2017bounding}, which is $O((d^4T^2/\epsilon^2+\|f\|_{\max})T)$, the above analysis polynomially improves the dependence on $1/\epsilon$ and polynomially reduces the dependence on $T$ and $d$. However, the factor of $5^{2k}$ in the gate complexity suggests that it may not be practical to apply the method for large values of $k$.

\subsection{Hamiltonian simulation by truncated Taylor series}

We now consider using the truncated Taylor series algorithm \cite{berry2015simulating} to simulate \eqn{semi_discrete_ODE_time} within error $\epsilon/2$. We improve upon the result of Ref.~\cite{kivlichan2017bounding} (which also uses the truncated Taylor series method) by using an improved representation with less spatial discretization error.

First we describe the complexity of simulating the discretized Hamiltonian produced by the Fourier spectral method.

	\begin{lemma}[Truncated Taylor series for discretized simulation]\label{lem:time-independent-Hamiltonian-simulation}
		Consider an instance of time-independent Hamiltonian simulation as defined in \eqn{spectral_Hamiltonian_time}, with a time-independent potential $f(\vect{x})$ satisfying $\|f(\vect{x})\|_{L^{\infty}} \le \|f\|_{\max}$ and a given $T>0$. Let $g' = \max_t\|\Phi^{(n/2)}(\cdot,t)\|_{L^1}$ as in \eqn{gprime}. There exists a quantum algorithm producing a normalized state that approximates $|\widetilde \Phi(T)\rangle$ with $\ell_2$ error at most $\epsilon/2$, with asymptotic gate complexity
		\begin{align}\label{eqn:time-independent-Hamiltonian-simulation}
            d(d+\|f\|_{\max})T\poly\Bigl(\log(dTg'/\epsilon)\Bigr).
		\end{align}
	\end{lemma}

	\begin{proof}
Let $C_A$ and $C_B$ denote the cost of querying the sparse Hamiltonian oracle for Hermitian matrices $A$ and $B$, respectively, and let $\alpha_A$ and $\alpha_B$ upper bound $\|A\|$ and $\|B\|$, respectively. Then the gate complexity of performing the simulation $e^{-i(A+B)T}$ is \cite{berry2015simulating}
		\begin{align}
(C_A+C_B)(\alpha_A+\alpha_B)T\frac{\log((\alpha_A+\alpha_B)T/\epsilon)}{\log(\log((\alpha_A+\alpha_B)T/\epsilon))}.
		\end{align}

		To simulate \eqn{spectral_Hamiltonian_time}, we take $A=\vect{F}^s_{n,d}\vect{L}_{n,d}(\vect{F}^s_{n,d})^{-1}$ and $B={\vect{V}}_{n,d}$.
		Using $\|\vect{L}_{n,d}\|\le\frac{dn^2}{4}$, $\|\vect{V}_{n,d}\|\le \|f\|_{\max}$, and the fact that the gate complexity of performing each of $\vect{L}^s_{n,d}$, $\vect{V}^s_{n,d}$, $\vect{F}^s_{n,d}$, or $(\vect{F}^s_{n,d})^{-1}$ is $d\poly(\log n)$, we obtain the gate complexity
		\begin{align}
			d\poly(\log n)(dn^2+\|f\|_{\max})T\frac{\log((dn^2+\|f\|_{\max})T/\epsilon)}{\log(\log((dn^2+\|f\|_{\max})T/\epsilon))}.
		\end{align}
		Using the value of $n$ from \eqn{truncated-number}, we see that the complexity is
		\begin{align}
			d(d+\|f\|_{\max})T\poly\Bigl(\log(dTg'/\epsilon)\Bigr)
		\end{align}
		as claimed.
	\end{proof}

	\begin{theorem}[Truncated Taylor series for real-space simulation]\label{thm:time-independent-Schrodinger-simulation}
		Consider an instance of the Schr{\"o}dinger equation \eqn{Schrodinger-indep}, with a time-independent potential $f(\vect{x})$ satisfying $\|f(\vect{x})\|_{L^{\infty}} \le \|f\|_{\max}$ and a given $T>0$. Let $g' = \max_t\|\Phi^{(n/2)}(\cdot,t)\|_{L^1}$ as in \eqn{gprime}. There exists a quantum algorithm producing a normalized state that approximates $\Phi(\vect{x},T)$ at the nodes $\{\vect{\chi}_{\vect{l}}\}$ defined as \eqn{interpolation_nodes}, with $\ell_2$ error at most $\epsilon$, with the gate complexity
		\begin{align}\label{eqn:time-independent-Schrodinger-simulation}
			d(d+\|f\|_{\max})T\poly\Bigl(\log(dTg'/\epsilon)\Bigr).
		\end{align}
	\end{theorem}

	\begin{proof}
		The result follows immediately from the same logic as in the proof of \thm{time-independent-Schrodinger-simulation-SO}.
	\end{proof}

Whereas the gate complexity in \cite{kivlichan2017bounding} is
	$O((d^4T^2/\epsilon^2+\|f\|_{\max})T)$, our approach achieves complexity $\widetilde O(d(d+\|f\|_{\max})T\log(1/\epsilon))$ in terms of $\ell_2$ error, exponentially improving the dependence on $1/\epsilon$, reducing the cubic dependence on $T$ to linear, and reducing the quartic dependence on $d$ to quadratic.

	\subsection{Interacting multi-particle systems}\label{sec:multi-particles}

	Now we consider simulating a multi-particle Schr{\"o}dinger equation
	\begin{align}\label{eqn:many-Schrodinger-dep}
		i\frac{\partial}{\partial t}\Phi(\vect{x},t)=\Big[-\frac{1}{2}\nabla^{2}+f(\vect{x},t)\Big]\Phi(\vect{x},t)
	\end{align}
	with a fixed number of particles $\eta$ in $d$ dimensions, interacting through a potential function $f(x)$.
	Here $\vect{x}\in\R^{\eta d}$ represents the positions of the particles, where entries $x_{(j-1)d+1}, \ldots, x_{jd}$ indicate the position of particle $j\in[\eta]$, $t\in\R$ represents time, $\nabla=\left(\frac{\partial}{\partial x_{i}}\right)\Big|_{i=1}^{\eta d}$, and $f\colon\R^{\eta d}\times\R\to\R$ is the potential function of the Schr{\"o}dinger equation.
	As above, we consider the case where $f$ is independent of time in this section.
	Also, we assume for simplicity that all particles have the same mass; this is easily generalized to the case of particles with different masses, as discussed in \sec{qchem}.

To make simulation tractable, we consider an $\eta d$-dimensional hypercubic domain $\Omega = [0,1]^{\eta d}$ and assume the wave function can be treated as periodic on this domain \cite{kivlichan2017bounding,martin2004electronic,su2021fault,babbush2017low}. The periodic boundary condition is natural for crystalline solids. As for a non-periodic system subject to a long-range potential such as the Coulomb potential, we can embed the system into a sufficiently large periodic box $\Omega$ such that the particles remain far from the boundary. Then, we can implement quantum simulations within the periodic box $\Omega$ because the tail of the wave function outside of $\Omega$ is negligible. In this case, one can imagine the full space $\R^{\eta d}$ is covered by repeated copies of the potential function restricted to $\Omega$, but the periodic images of the potential outside of the box $\Omega$ do not significantly interact with the wave functions supported on $\Omega$ \cite{martin2004electronic,su2021fault,babbush2017low}. In practice, it is not necessary to simulate the periodic images of the potential function outside of $\Omega$.

We also want the potential function $f(\vect{x})$ to be bounded, i.e., $\|f(\vect{x})\|_{L^{\infty}} \le \|f\|_{\max}$. However, typical potentials arising in physics include singularities, such as the divergence of the Coulomb potential for two particles at the same location. We can handle this by modifying the potential in a way that does not significantly affect the solution at relevant length scales.  For example, a $d$-dimensional generalization of the Coulomb potential can be modified as \cite{kivlichan2017bounding}
	\begin{align}\label{eqn:modified-Coulomb}
		f_{\text{Coulomb}}(\vect x) = \sum_{1\le i<j\le\eta} \frac{q_iq_j}{\sqrt{\sum_{k=1}^d\bigl(x_{(i-1)d+k}-x_{(j-1)d+k}\bigr)^2+\Delta^{2}}},
	\end{align}
	where $q_i$ is the charge of the $i$th particle and $\Delta>0$ serves to keep the potential bounded \cite{kivlichan2017bounding}. Letting $q:=\max_i|q_i|$, we have
	\begin{align}\label{eqn:modified-Coulomb-bound}
		\|f(\vect{x})\|_{L^{\infty}} \le \frac{\eta(\eta-1)q^2}{2\Delta} = \|f\|_{\max}.
	\end{align}

    The parameter $\Delta$ captures how closely the modified Coulomb potential \eqn{modified-Coulomb} approximates the unbounded potential. To accurately reproduce the behavior of the unbounded Coulomb potential, we would like to simulate the model for small $\Delta>0$, and we expect the complexity of the simulation to grow with $1/\Delta$ as a consequence of the upper bound \eqn{modified-Coulomb-bound}. In practice, the modified Coulomb potential \eqn{modified-Coulomb} should give a good approximation of the original Coulomb potential provided particles remain separated by distances large compared with $\Delta$.

As considered in~\cite{berry2015simulating,kivlichan2017bounding,low2018hamiltonian}, we analyze the gate complexity of implementing the sparse Hamiltonian oracle, where we count a query to the modified Coulomb potential in the implementation as one gate.

\thm{time-independent-Schrodinger-simulation-SO} and \thm{time-independent-Schrodinger-simulation} directly imply quantum algorithms for simulating \eqn{many-Schrodinger-dep} using product formulas and the truncated Taylor series method, respectively.

	\begin{theorem}[$k$th-order product formula simulation of interacting particles]\label{thm:time-independent-many-Schrodinger-simulation-SO}
		Consider an instance of the multi-particle Schr{\"o}dinger equation \eqn{many-Schrodinger-dep} with a time-independent potential $f(\vect{x})$ satisfying $\|f(\vect{x})\|_{L^{\infty}} \le \|f\|_{\max}$, and a given $T>0$. Let $g' = \max_t\|\Phi^{(n/2)}(\cdot,t)\|_{L^1}$ as in \eqn{gprime}. There exists a quantum algorithm producing a normalized state that approximates $\Phi(\vect{x},T)$ at the nodes $\{\vect{\chi}_{\vect{l}}\}$ defined in \eqn{interpolation_nodes}, with $\ell_2$ error at most $\epsilon$, with asymptotic gate complexity
		\begin{align}\label{eqn:time-independent-many-Schrodinger-simulation-SO}
            \widetilde O\Bigl( 5^{2k} \eta d(\eta d+\|f\|_{\max})^{1+1/2k}T^{1+1/2k}/\epsilon^{1/2k} \Bigr).
		\end{align}
	\end{theorem}

	\begin{proof}
		It suffices to replace $d$ by $\eta d$ in the proof of \thm{time-independent-Schrodinger-simulation-SO}.
	\end{proof}

	\begin{theorem}[Truncated Taylor series simulation of interacting particles]\label{thm:time-independent-many-Schrodinger-simulation}
		Consider an instance of the multi-particle Schr{\"o}dinger equation \eqn{many-Schrodinger-dep}, with a time-independent potential $f(\vect{x})$ satisfying $\|f(\vect{x})\|_{L^{\infty}} \le \|f\|_{\max}$, and a given $T>0$. Let $g' = \max_t\|\Phi^{(n/2)}(\cdot,t)\|_{L^1}$ as in \eqn{gprime}. There exists a quantum algorithm producing a normalized state that approximates $\Phi(\vect{x},T)$ at the nodes $\{\vect{\chi}_{\vect{l}}\}$ defined in \eqn{interpolation_nodes}, with $\ell_2$ error at most $\epsilon$, with asymptotic gate complexity
		\begin{align}\label{eqn:time-independent-many-Schrodinger-simulation}
            \eta d(\eta d+\|f\|_{\max})T\poly\bigl(\log(\eta dTg'/\epsilon)\bigr).
		\end{align}
	\end{theorem}

	\begin{proof}
		As in the previous result, it suffices to replace $d$ by $\eta d$ in the proof of \thm{time-independent-Schrodinger-simulation}.
	\end{proof}
	

	\section{Simulating time-dependent Schr{\"o}dinger equations}\label{sec:time-dep}

So far, we have focused on quantum algorithms for simulating systems with time-independent potentials. However, we saw in \sec{fourier-spectral} that the Fourier spectral method can be readily applied to time-dependent Schr\"odinger equations, retaining exponential convergence. Thus the quantum simulation problem for time-dependent potentials effectively reduces to a time-dependent Hamiltonian simulation problem with a discretized Hamiltonian \eqn{spectral_Hamiltonian_time} of the form $H(t) = A + B(t)$. In this section, we apply known methods for simulating time-dependent Hamiltonians \cite{low2018hamiltonian, berry2020time} to give concrete bounds on the complexity of simulating time-dependent Schr\"odinger equations in real space.

\subsection{Review of time-dependent Hamiltonian simulation methods}

	We take a detour from the real-space simulation problem to motivate the two main techniques used in this section: Hamiltonian simulation in the interaction picture \cite{low2018hamiltonian} and the rescaled Dyson-series algorithm \cite[Section 4]{berry2020time}.

	Suppose we want to simulate a time-dependent Hamiltonian of the form $H(t) = A + B(t)$. The quantum state evolves as
	\begin{align}
		\ket{\psi(t)} = \mathcal{T}\left[e^{-i\int^t_0 H(s) \d s}\right] \ket{\psi(0)},
	\end{align}
	where $\mathcal{T}$ is the time-ordering operator, so that $\mathcal{T}\left[e^{-i\int^t_0 H(s) \d s}\right] = \lim_{r\to \infty} \prod^r_{j=1} e^{-iH(jt/r) t/r}$.

	Define the interaction-picture Hamiltonian
	\begin{align}\label{eqn:interaction-hamiltonian}
		H_I(t) := e^{iAt} B(t) e^{-iAt}
	\end{align}
	and $\ket{\psi_I(t)} := e^{iAt}\ket{\psi(t)}$ for all $t$. Moving to the interaction picture may be advantageous since $\|H_I(t)\| = \|B(t)\| \le \|H(t)\|$. One can easily check that
	\begin{align}\label{eqn:interaction-schrodinger-eq}
		i \partial_t \ket{\psi_I(t)} = H_I(t) \ket{\psi_I(t)}.
	\end{align}

	In \cite{low2018hamiltonian}, the time-dependent Hamiltonian simulation problem \eqn{interaction-schrodinger-eq} is addressed using the Dyson-series technique, giving query and gate complexity that scales with the $\max$-norm of the interaction Hamiltonian $H_I(t)$ \cite[Section 5]{low2018hamiltonian}. This can be improved to scale with the $L^1$-norm of $H_I(t)$ using the rescaled Dyson-series algorithm \cite[Section 4]{berry2020time}.

	For a time-dependent Hamiltonian $H(t)$, we define the rescaled Hamiltonian
	\begin{align}
		\widetilde{H}(\varsigma) := \frac{H(g^{-1}(\varsigma))}{\|H(g^{-1}(\varsigma))\|_{\max}},
	\end{align}
	where
	\begin{align}\label{eqn:g-function}
		g(t) := \int^t_0 \|H(s)\|_{\max} ~\d s
	\end{align}
	is the $L^1$-norm of $\|H(t)\|_{\max}$. Further, we define
	\begin{align}\label{eqn:max-1}
	    \|H\|_{\max,1}:=\int_{0}^{T}\|H(t)\|_{\max}\d t=g(T),
	\end{align}
    where $T$ denotes the total simulation time. A key observation is that the rescaling of the Hamiltonian does not affect the target state:
	\begin{align}
		\ket{\psi(t)} = \mathcal{T}\left[e^{-i\int^t_0 H(s) \, \d s}\right] \ket{\psi(0)} = \mathcal{T}\left[e^{-i\int^{g(t)}_0 \widetilde{H}(\varsigma) \, \d \varsigma}\right] \ket{\psi(0)}.
	\end{align}	
	Given this rescaling procedure, if we have
	\begin{enumerate}
		\item an algorithm that simulates the rescaled Hamiltonian $\tilde{H}(\varsigma)$ for $0 \le \varsigma \le T$ with $L^\infty$-norm cost, i.e., with complexity $O(T\|\tilde{H}(\varsigma)\|_{\max, \infty})$ where
		\begin{align}
			\|\tilde{H}(\varsigma)\|_{\max, \infty} = \sup_{\varsigma \in [0,T]}\|\tilde{H}(\varsigma)\|_{\max},
		\end{align}
		and
		\item the ability to compute $g^{-1}(\varsigma)$ for any $\varsigma \in [0,T]$ and the max-norm $\|H(s)\|_{\max}$ for any $s \in [0,T]$ (so that we have access to the rescaled Hamiltonian $\tilde{H}(\varsigma)$ for any $\varsigma$),
	\end{enumerate}
	then we are able to simulate the original Hamiltonian $H(T)$ for $0 \le t \le T$ with $L^1$-norm cost, i.e., with complexity $O(\int^T_0 \|H(t)\|_{\max}~\d t)$. To see this, note that $\|\tilde{H}(\varsigma)\|_{\max} = 1$ for all $\varsigma \in [0, g(T)]$. Therefore, if we apply the simulation algorithm with $L^\infty$-norm cost to $\tilde{H}(\varsigma)$ for $0 \le \varsigma \le g(T)$, the cost is bounded by
	$g(T)$, which is the the $L^1$-norm of $\|H(s)\|_{\max}$:
	\begin{align}
		g(T) \|\tilde{H}(\varsigma)\|_{\max, \infty} = g(T),
	\end{align}
	which is the $L^1$ norm.

	Now, we apply the above rescaling procedure to the interaction-picture Hamiltonian $H_I(t)$ \eqn{interaction-hamiltonian}. The rescaled interaction-picture Hamiltonian $\widetilde{H}_{I}(\varsigma)$ satisfies
	\begin{align}\label{eqn:interaction-picture-hamt}
		\widetilde{H}_{I}(\varsigma)= e^{iAg^{-1}(\varsigma)} \widetilde{B}(\varsigma) e^{-iAg^{-1}(\varsigma)},
	\end{align}
	where \begin{align}\label{eqn:tilde-B-def}
	    \widetilde{B}(\varsigma):=B(g^{-1}(\varsigma))/\|B(g^{-1}(\varsigma))\|_{\max}
	\end{align}
	stands for the rescaled operator for $B$.

	\begin{lemma}
		The max-norm of the rescaled interaction-picture Hamiltonian $\widetilde{H}_{I}(\varsigma)$ as defined in \eqn{interaction-picture-hamt} is bounded by $1$ for any $0 \le \varsigma \le g(t)$:
		\begin{align}
			\|\widetilde{H}_{I}(\varsigma)\|_{\max} \le 1.
		\end{align}
	\end{lemma}
	\begin{proof}
		First, note that the max-norm of any matrix is upper bounded by its $\ell_{2}$-norm, we have
		\begin{align}
			\|\widetilde{H}_{I}(\varsigma)\|_{\max} \le \|\widetilde{H}_{I}(\varsigma)\| \le \|e^{iAg^{-1}(\varsigma)}\| \|\widetilde{B}(\varsigma)\| \|e^{-iAg^{-1}(\varsigma)}\| \le \|\widetilde{B}(\varsigma)\|.
		\end{align}
		The last two steps hold because the matrix $\ell_2$-norm is sub-multiplicative: $\|AB\|\le \|A\|\|B\|$, while the unitary operators $e^{iAg^{-1}(\varsigma)}$ and $e^{-iAg^{-1}(\varsigma)}$ have unit $\ell_{2}$-norm. Since the rescaled operator $\widetilde{B}(\varsigma)$ is diagonal, its max-norm and $\ell_{2}$-norm are equally $1$. This proves our claim that $\|\widetilde{H}_{I}(\varsigma)\|_{\max} \le 1$.
	\end{proof}

	Similarly, if we have an $L^\infty$-norm algorithm (e.g., the standard Dyson-series simulation) that simulates the rescaled interaction-picture Hamiltonian $\widetilde{H}_{I}(\varsigma)$ for $0 \le \varsigma \le g(t)$ as well as the access to rescaled Hamiltonian $\widetilde{H}_{I}(\varsigma)$, we can simulate the original Hamiltonian $H(s)$ with the total cost $O(g(t)) = O(\int^t_0\|B(s)\|_{\max}~\d s)$. Note that although the Hamiltonian is a sum of two operators, $H(t)= A+B(t)$, the simulation cost only depends on the $L^1$-norm of $\|B(s)\|_{\max}$ and not $A$. This is the advantage of using the interaction picture simulation method.

	\subsection{Block-encoding of the rescaled interaction-picture Hamiltonian}

	The input model of the truncated Dyson series method is a unitary oracle for a so-called block-encoding (defined in \append{notation}). In the present context, the values of the potential $f(\vect{x}, t)$ can be produced by an evaluation oracle \eqn{evaluation-dep}. Therefore, we provide an explicit construction of the relevant block-encoding oracle using queries to the potential.

	\paragraph{Efficient simulations of $e^{iAt}$.} Interaction picture simulation is adventageous when it is easy to implement $e^{iAt}$, even for large $t$. When simulating the time-dependent Schr\"odinger equation \eqn{Schrodinger-dep} with the Fourier spectral method, the $A$ term in the Hamiltonian \eqn{spectral_Hamiltonian_time} is
	\begin{align}\label{eqn:A-term}
		A = \vect{F}^s_{n,d}{\vect{L}}_{n,d}(\vect{F}^s_{n,d})^{-1},
	\end{align}
	where ${\vect{L}}_{n,d}$ is an explicit diagonal matrix and $\vect{F}^s_{n,d}$ is the multi-dimensional quantum shifted Fourier transform (QSFT). The transformation $\vect{F}^s_{n,d}$ and its inverse can be performed with gate complexity ${O}(d\log n\log\log n)$ \cite[Lemma 5]{childs2020high}. Therefore, $e^{iAt}$ can be simulated as $\vect{F}^s_{n,d} e^{i{\vect{L}}_{n,d} t} (\vect{F}^s_{n,d})^{-1}$. By standard techniques (see for example Rule 1.6 in \cite[Section 1.2]{childs2004quantum}), $e^{i{\vect{L}}_{n,d}t}$ can be simulated with gate complexity $\tilde{O}(d\log n)$.
\begin{lemma}\label{lem:U-tilde-B}
		Let $B\in\C^{2^{w}\times2^{w}}$ be a diagonal matrix. Suppose we have access to the evaluation oracle $U_{f}$ that returns the values of diagonal entries of $B$ in binary as
		\begin{align}
			U_{f}\ket{j}\ket{t}\ket{0^{z}} = \ket{j}\ket{t}\ket{f_{j}(t)}=\ket{j}\ket{t}\ket{B_{jj}(t)},\quad\forall j\in[2^{w}]-1,
		\end{align}
		where $B_{jj}=f_{j}(t)$ is a $z$-bit binary description of the $j$th diagonal entry of $B(t)$. Additionally, suppose we have the following two oracles\footnote{This is a standard assumption; see for example~\cite[Section 4.2]{berry2020time}.} that implement the inverse change-of-variable and compute the max-norm:
\begin{align}
	O_{\inv}\ket{\varsigma,z}&=\ket{\varsigma,z \oplus g^{-1}(\varsigma)},\\
	O_{\nrm}\ket{\tau,z}&=\ket{\tau, z\oplus\|B(\tau)\|_{\max}}.
\end{align}
		Then we can implement the following evaluation oracle $U_{\widetilde{B}}$ for the rescaled operator $\widetilde{B}(\varsigma)$ defined in \eqn{tilde-B-def}, acting as
		\begin{align}
			U_{\widetilde{B}}\ket{j}\ket{\varsigma}\ket{0}^{z}=\ket{j}\ket{\varsigma}\ket{\widetilde{B}_{jj}(\varsigma)}\qquad\forall j\in[2^{w}]-1,
		\end{align}
	using $O(1)$ queries to the oracles $U_{f}$, $O_{\inv}$, $O_{\nrm}$ and additionally using $O(z^{2})$ one- and two-qubit gates.
	\end{lemma}

	\begin{proof}
		The main idea follows the proof of \cite[Theorem 10]{berry2020time}. Let the function $g(t)$ be defined as in \eqn{g-function}.

		We use oracles $O_{\inv}$ and $O_{\nrm}$ to implement the transformation
		\begin{align}
			\ket{\varsigma,0,0}\mapsto\ket{\varsigma,g^{-1}(\varsigma),\|B(g^{-1}(\varsigma))\|_{\text{max}}}.
		\end{align}
		We then query $U_{f}$ and normalize the result with $\|B(g^{-1}(\varsigma))\|_{\text{max}}$ to compute $\widetilde{B}_{jj}(\varsigma)$:
		\begin{align}
			\ket{g^{-1}(\varsigma),\|B(g^{-1}(\varsigma))\|_{\text{max}},j,0}\mapsto\ket{g^{-1}(\varsigma),\|B(g^{-1}(\varsigma))\|_{\text{max}},j,\widetilde{B}_{jj}(\varsigma)}.
		\end{align}
		Then we uncompute the ancilla registers and obtain an evaluation oracle for the rescaled operator, namely
		\begin{align}
			U_{\widetilde{B}}\ket{j}\ket{\varsigma}\ket{0}^{z}=\ket{j}\ket{\varsigma}\ket{\widetilde{B}_{jj}(\varsigma)}\qquad\forall j\in[2^{w}]-1.
		\end{align}
		This process uses $O(1)$ queries to the oracles $U_{f}$, $O_{\inv}$, and $O_{\nrm}$.

		We now analyze the gate complexity. Each entry of $B$ is given using $z$ qubits. The above implementation process only involves arithmetic operations on these qubits, which can be implemented with complexity $O(z^{2})$. Such operation is called only if an oracle query happens, so we can implement all the arithmetic operations with $O(z^{2})$ additional gates.
	\end{proof}

	Next, we evaluate the cost of implementing the unitary oracle $\hamt_{\widetilde{H}_{I}}$ of the rescaled interaction-picture Hamiltonian $\widetilde{H}_{I}(\varsigma)$. The definition of the $\hamt$ oracle, originally from \cite{low2018hamiltonian}, can be found in \append{notation}.

	\begin{lemma}\label{lem:hamt-hi}
		Let $\widetilde{H}_{I}$ be the rescaled interaction-picture Hamiltonian in \eqn{interaction-picture-hamt}. Then the oracle $\hamt_{\widetilde{H}_{I}}$ can be approximated with error at most $\delta$ at the following cost:
		\begin{enumerate}
			\item Queries to the oracles $U_{f}$, $O_{\inv}$ and $O_{\nrm}$: $O(1)$,
			\item One- and two-qubit gates: $\widetilde{O}\big(z^{2}+\log^{2.5}(1/\delta)+d\log n\big)$.
		\end{enumerate}
	\end{lemma}
	\begin{proof}
	    Observe that $\hamt_{\widetilde{H}_{I}}$ can be approximated by $e^{iAg^{-1}(\varsigma)} O_{\widetilde{B}} e^{-iAg^{-1}(\varsigma)}$ to any desired accuracy. Since we require the error of $\hamt_{\widetilde{H}_{I}}$ be bounded by $\delta$, the error of the above three terms should be $O(\delta)$.

	    By \lem{U-tilde-B} and \lem{block-encoding-diagonal}, an $(1,\omega+3,\epsilon)$-block-encoding of $\tilde{B}$, denoted $O_{\widetilde{B}}$, can be approximated. with $O(1)$ queries to oracles $U_{f}$, $O_{\text{inv}}$, $O_{\text{norm}}$, and $O(z^{2}+w+\log^{2.5}(1/\delta))$ one- and two-qubit gates.

	    To implement $e^{iAg^{-1}(\varsigma)}$ and $e^{-iAg^{-1}(\varsigma)}$, the diagonal elements and $g^{-1}(\varsigma)$ should be known to $O(\log (1/\delta))$ bits of precision. Then they can be simulated using $\widetilde{O}(d\log n+\log(1/\delta))$ one- and two-qubit gates. Note that $w$ is at most $O(d\log n)$, then in total $\widetilde{O}\big(z^{2}+\log^{2.5}(1/\delta)+d\log n\big)$ one- and two-qubit gates suffice.
	\end{proof}

	\subsection{Rescaled Dyson-series algorithm with \texorpdfstring{$L^1$}{L1}-norm scaling}

The rescaled Dyson-series algorithm uses a block-encoding oracle to achieve the following simulation \cite{low2018hamiltonian}.

	\begin{lemma}[{\cite[Lemma 6]{low2018hamiltonian}}]\label{lem:interaction-picture-query-complexity}
		Let $A\in\C^{2^{n_{s}}\times 2^{n_{s}}}$, $B\in\C^{2^{n_{s}}\times 2^{n_{s}}}$, and let $\alpha_{A},\alpha_{B}$ be known constants such that $\|A\|\leq \alpha_{A}$ and $\|B\|\leq\alpha_{B}$. Assume the existence of a unitary oracle $\hamt_{\widetilde{H}_{I}}$ that block-encodes the Hamiltonian within the interaction picture, which implicitly depends on the time-step size $\tau=O(\alpha_{B}^{-1})$ and number of time-steps $M=O(t(\alpha_{A}+\alpha_{B})/\epsilon)$. Then for all $t\geq2\alpha_{B}\tau$, the time-evolution operator $e^{-i(A+B)t}$ may be approximated to error $\epsilon$ with the following cost.
		\begin{enumerate}
			\item Simulations of $e^{-iA\tau}$: $O(\alpha_{B}t)$,
			\item Queries to $\hamt_{\widetilde{H}_{I}}$: $O\big(\alpha_{B}t\frac{\log(\alpha_{B}t/\epsilon)}{\log\log(\alpha_{B}t/\epsilon)}\big)$,
			\item Primitive gates: $O\Big(\alpha_{B}t\big(n_a+\log(t(\alpha_{A}+\alpha_{B})/\epsilon)\big)\alpha_{B}t\frac{\log(\alpha_{B}t/\epsilon)}{\log\log(\alpha_{B}t/\epsilon)}\Big)$.
		\end{enumerate}
	\end{lemma}
	When simulating a time-dependent Schr{\"o}dinger equation, it is natural to let $B$ be the term corresponding to the potential energy $f(\vect{x},t)$. Then $B$ is a diagonal matrix and we have access to its diagonal components through the evaluation oracle \eqn{evaluation-dep}. This can be used to efficiently implement the evaluation oracle for its rescaling $\widetilde{B}$, which can in turn be used to implement the Hamiltonian $\hamt_{\widetilde{H}_{I}}$ specified in \defn{hamt}.

	\begin{proposition}\label{prop:rescaled-interaction-complexity}
		For $\tau\in[0,T]$, let $H(\tau)=-\nabla^{2}+B(\tau)$ be a discretized real-space Hamiltonian with $B(\tau)$ a diagonal matrix whose diagonal elements are specified by the potential function $f(\vect{x},\tau)\colon\R^{d}\times\R\to\R$. Suppose that $f$ is a positive and continuously differentiable function, $L$-Lipschitz in terms of $\tau$, and can be computed with $z$ bits of precision. Then $H$ can be simulated for time $T$ with accuracy $\epsilon$ with the following costs.
		\begin{enumerate}
			\item Queries to $U_{f}$, $O_{\inv}$, $O_{\nrm}$: $O\Big(\|f\|_{\max,1}\frac{\log(\|f\|_{\max,1}/\epsilon)}{\log\log(\|f\|_{\max,1}/\epsilon)}\Big)$ where $\|\cdot\|_{\max,1}$ is defined in \eqn{max-1} and the oracles are defined in \lem{U-tilde-B};
			\item One- and two-qubit gates:\\
			$O\bigg(\|f\|_{\max,1}\Big(\poly(z)+\log^{2.5}(L\|f\|_{\max,1}/\epsilon)+d\log \Big(\frac{\log(g'/\epsilon)}{\log\log(g'/\epsilon)}\Big)\Big)\frac{\log(\|f\|_{\max,1}/\epsilon)}{\log\log(\|f\|_{\max,1}/\epsilon)}\bigg)$;
		\end{enumerate}
	    where $g' =  \max_t\|\Phi^{(n/2)}(\cdot,t)\|_{L^1}$ as defined in \eqn{gprime}.
	\end{proposition}
	The following lemma from \cite{low2018hamiltonian} is useful for proving this theorem:
	\begin{lemma}[{\cite[Corollary 4]{low2018hamiltonian}}]\label{lem:hamt-query-complexity}
		Let $\widetilde{H}(\varsigma): [0,t]\rightarrow\C^{2^{w}\times2^{w}}$, and suppose $\|\widetilde{H}\|_{\max}$ is bounded by some constant $C_{\widetilde{H}}$, and $\left<\|\frac{\d}{\d\varsigma}\widetilde{H}\|\right>=\frac{1}{t}\int_{0}^{t}\|\frac{\d}{\d\varsigma}\widetilde{H}(\varsigma)\|\d\varsigma$. Set $M=O\Big(\frac{\|\tilde{H}\|_{\max,1}}{\epsilon}\Big(\left<\|\frac{\d}{\d\varsigma}\widetilde{H}\|\right>+\|\widetilde{H}\|_{\max}^{2}\Big)\Big)$ to be the number of time intervals in the $\hamt_{\widetilde{H}}$ oracle. Then for all $\epsilon>0$, an operation $W$ can be implemented with failure probability at most $O(\epsilon)$, such that
		\begin{align}
			\|W-\mathcal{T}\big[e^{\int_{\varsigma=0}^{g^{-1}(t)}\widetilde{H}(\varsigma) \d \varsigma}\big]\| \leq \epsilon
		\end{align}
		with the following costs.
		\begin{itemize}
			\item[1.] Queries to $\hamt_{\widetilde{H}}$: $O\Big(\|\widetilde{H}\|_{\max,1}\frac{\log(\|\widetilde{H}\|_{\max,1}/\epsilon)}{\log\log(\|\widetilde{H}\|_{\max,1}/\epsilon)}\Big)$,
			\item[2.] Primitive gates: $O\Big(\|\widetilde{H}\|_{\max,1}\Big(n_{a}+\log\Big(\frac{\|\widetilde{H}\|_{\max,1}}{\epsilon}\Big(\left<\|\frac{\d}{\d\varsigma}\widetilde{H}\|\right>+\|\widetilde{H}\|_{\max}^{2}\Big)\Big)\Big)\frac{\log(\|\widetilde{H}\|_{\max,1}/\epsilon)}{\log\log(\|\widetilde{H}\|_{\max,1}/\epsilon)}\Big)$.
		\end{itemize}
		Here $n_{a}$ denotes the number of ancillary qubits needed to implement $\hamt_{\widetilde{H}}$.
	\end{lemma}
	Now we present the proof of \prop{rescaled-interaction-complexity}.
	\begin{proof}
		We follow the simulation method of \cite{low2018hamiltonian}. Whenever $\hamt_{\widetilde{H}_{I}}$ is called to obtain the block-encoding of $\widetilde{H}_{I}(\varsigma)$, since in the interaction picture the simulation time equals $\|f\|_{\max,1}$, we require an $O(\epsilon/\|f_{\max,1}\|)$-approximate $\hamt_{\widetilde{H}_{I}}$ to guarantee that the overall error is bounded by $\epsilon$. By \lem{hamt-hi}, this can be achieved using $O(1)$ queries to oracles $U_{f}$, $O_{\inv}$, $O_{\nrm}$ and $\widetilde{O}\big(z^{2}+\log^{2.5}(\|f\|_{\max,1}/\epsilon)+d\log n\big)$ one- and two-qubit gates.

        Also note that
		\begin{align}		    \|\widetilde{H}_{I}(\varsigma)\|_{\max}\leq\|\widetilde{H}_{I}(\varsigma)\|=\|\widetilde{B}(\varsigma)\|=\|\widetilde{B}(\varsigma)\|_{\max}\ \forall \varsigma,
		\end{align}
		so
		\begin{align}
			\|\widetilde{H}_{I}\|_{\max,1}\leq \|B\|_{\max,1}=\|f\|_{\max,1}.
		\end{align}
		Then by \lem{hamt-query-complexity}, the query complexity to oracles $U_{f}$, $O_{\inv}$ and $O_{\nrm}$ is
		\begin{align}
			O\Big(\|f\|_{\max,1}\frac{\log(\|f\|_{\max,1}/\epsilon)}{\log\log(\|f\|_{\max,1}/\epsilon)}\Big).
		\end{align}

		We have
		\begin{align}
			\frac{\d \widetilde{H}_{I}}{\d \varsigma}=e^{iAg^{-1}(\varsigma)}\frac{\d\widetilde{B}}{\d\varsigma}e^{-iAg^{-1}(\varsigma)}+i\cdot\frac{\d g^{-1}}{\d\varsigma}(A e^{iAg^{-1}(\varsigma)}\widetilde{B}e^{-iAg^{-1}(\varsigma)}+e^{iAg^{-1}(\varsigma)}\widetilde{B}Ae^{-iAg^{-1}(\varsigma)}),
		\end{align}
		and since $A=-\nabla^{2}$ is a Hermitian operator, we can further deduce that
		\begin{align}
			\frac{\d \widetilde{H}}{\d \varsigma}=e^{iAg^{-1}(\varsigma)}\frac{\d\widetilde{B}}{\d\varsigma}e^{-iAg^{-1}(\varsigma)},
		\end{align}
        where
		\begin{align}
			\Big\|\frac{\d \widetilde{H}_{I}}{\d \varsigma}\Big\|=\Big\|\frac{\d\widetilde{B}}{\d\varsigma}\Big\|=O\Big(\Big\|\frac{\d B}{\d t}\Big\|\Big)=O(L).
		\end{align}
		Hence, the number of one- and two-qubit gates needed in this procedure is
		\begin{align}
			&O\Big(\|f\|_{\max,1}\Big(n_{a}+\log(L\|f\|_{\max,1}/\epsilon)\Big)\frac{\log(\|f\|_{\max,1}/\epsilon)}{\log\log(\|f\|_{\max,1}/\epsilon)}\Big)\nonumber\\
	    	&\qquad+O\Big(\|f\|_{\max,1}(z^{2}+\log^{2.5}(\|f\|_{\max,1}/\epsilon)+d\log n)\cdot\frac{\log(\|f\|_{\max,1}/\epsilon)}{\log\log(\|f\|_{\max,1}/\epsilon)}\Big).
		\end{align}
		This expression uses the fact that each primitive gate can be implemented using $O(1)$ one- and two-qubit gates. By the proof of {\cite[Lemma 8]{low2018hamiltonian}}, $n_{a}$ can be upper bounded by $\poly(z)$. Furthermore, by \eqn{truncated-number}, the truncation parameter is $n = 2\biggl\lceil\frac{\log(\omega)}{\log(\log(\omega))}\biggr\rceil$ for $\omega=\frac{8g'}{\pi\epsilon}$. Therefore, the gate complexity can be expressed as
		\begin{align}
		\hspace{-4mm} O\bigg(\|f\|_{\max,1}\Big(\poly(z)+\log^{2.5}\Big(\frac{L\|f\|_{\max,1}}{\epsilon}\Big)+d\log \Big(\frac{\log(g'/\epsilon)}{\log\log(g'/\epsilon)}\Big)\Big)\frac{\log(\|f\|_{\max,1}/\epsilon)}{\log\log(\|f\|_{\max,1}/\epsilon)}\bigg)\!\!
		\end{align}
		as claimed.
	\end{proof}


    \subsection{Generalization to multi-particle systems}

    We now generalize \prop{rescaled-interaction-complexity} to the case of a fixed number of particles $\eta$ in $d$ dimensions interacting through a potential function $f(\vect{x})$. As in \sec{multi-particles}, here $\vect{x}\in\mathbb{R}^{\eta d}$ represents the positions of the particles and the evolution of the wave function $\Phi(\vect{x},t)$ follows the time-dependent multi-particle Schr{\"o}dinger equation \eqn{many-Schrodinger-dep}. In this subsection, we suppose the oracles $U_{f}$, $O_{\inv}$, $O_{\nrm}$ defined in \lem{U-tilde-B} are still available in this multi-particle scenario.
    Using \prop{rescaled-interaction-complexity}, the complexity of simulating \eqn{many-Schrodinger-dep} in the interaction picture is as follows.

    \begin{theorem}\label{thm:many-rescaled-interaction-complexity}
	For $\tau\in[0,T]$, let $H(\tau)=-\nabla^{2}+B(\tau)$ be a discretized multi-particle real-space Hamiltonian with $B(\tau)$ a diagonal matrix whose diagonal elements are specified by the potential function $f(\vect{x},\tau)$. Suppose that $f\colon\R^{\eta d}\times\R\to\R$ is a positive, continuously differentiable function, $L$-Lipschitz in terms of $\tau$, and can be computed with $z$ bits of precision. Then $H$ can be simulated for time $T$ with accuracy $\epsilon$ with the following cost:
	\begin{enumerate}
		\item Queries to $U_{f}$, $O_{\inv}$, $O_{\nrm}$: $O\Big(\|f\|_{\max,1}\frac{\log(\|f\|_{\max,1}/\epsilon)}{\log\log(\|f\|_{\max,1}/\epsilon)}\Big)$ where $\|\cdot\|_{\max,1}$ is defined in \eqn{max-1} and the oracles are defined in \lem{U-tilde-B};
		\item One- and two-qubit gates:\\
		$O\bigg(\|f\|_{\max,1}\Big(\poly(z)+\log^{2.5}(L\|f\|_{\max,1}/\epsilon)+\eta d\log \Big(\frac{\log(g'/\epsilon)}{\log\log(g'/\epsilon)}\Big)\Big)\frac{\log(\|f\|_{\max,1}/\epsilon)}{\log\log(\|f\|_{\max,1}/\epsilon)}\bigg)$;
	\end{enumerate}
	where $g' =  \max_t\|\Phi^{(n/2)}(\cdot,t)\|_{L^1}$ as defined in \eqn{gprime}.
	\end{theorem}

	\begin{proof}
	    It suffices to replace $d$ by $\eta d$ in the proof of \prop{rescaled-interaction-complexity} and choose the value of the truncation parameter $n$ as in \eqn{truncated-number} to obtain the complexity for general case.
	\end{proof}
	
	We now discuss the cost of simulating an $\eta$-electron system under the modified Coulomb potential $f(\vect{x})$ defined in \eqn{modified-Coulomb} using the approach of \thm{many-rescaled-interaction-complexity}. To give a complete algorithm, we must implement the evaluation oracle $U_f$ for the modified Coulomb potential. The most straightforward implementation directly computes the sum of $\eta(\eta-1)/2$ pairwise Coulomb interactions, giving gate complexity $O(\eta^2)$. However, more advanced numerical techniques for $\eta$-body problems such as the Barnes-Hut algorithm~\cite{barnes1986hierarchical} and the fast multipole method~\cite{greengard1987fast} can evaluate the $\eta$-particle Coulomb potential faster, reducing the cost to linear in $\eta$ when the dimension $d$ is a constant.

	In particular, for $d=3$ the Barnes-Hut algorithm~\cite{barnes1986hierarchical} proceeds as follows. We divide the unit cube into cubic cells in an octree structure of height $h$ until each cell contains at most one particle. We then compute the pairwise Coulomb interactions between particles in nearby cells. Finally, we approximate the remaining Coulomb interactions between particles in distant cells by treating nearby particles as a single large particle located at their geometric center. The cost of this approach is as follows.
	
	\begin{lemma}[\cite{sarin1998analyzing}]\label{lem:BH-tree}
	    For a system of $\eta$ particles in a three-dimensional unit cube,
	    the Barnes-Hut algorithm~\cite{barnes1986hierarchical} of height $h$ approximates the total Coulomb potential to a constant accuracy in time $O(\eta h)$. Adopting the fast multipole method~\cite{greengard1987fast} as a subroutine to estimate the Coulomb interaction between well-separated clusters of particles, the reulting multipole-based Barnes-Hut algorithm approximates the total Coulomb potential to accuracy $\epsilon$ in time $O(\eta h\log^2(1/\epsilon))$.
	\end{lemma}
	In particular, the error introduced by the fast multipole subroutine is bounded as follows.
	
	\begin{lemma}[\cite{greengard1987fast}]\label{lem:multipole-error}
	    Suppose that $k$ particles of charge $q_j$ for $j \in \{1,\ldots,k\}$ are located within a sphere centered at $0$ with radius $r_s$. Then after a reusable preprocessing step with time complexity $O(k)$, for any point $P$ at a distance $r>r_s$ from the origin, the $p$th-order fast multipole method approximates the Coulomb potential $\phi(P)$ to accuracy $\frac{Q}{r-r_s}\big(\frac{r_s}{r}\big)^{p+1}$ in time $O(p^2)$, where $Q:=\sum_{j=1}^k|q_j|$.
	\end{lemma}
	
	This method can be generalized to approximate the $d$-dimensional modified Coulomb potential defined in \lem{U-tilde-B}. We analyze the complexity of this method in a random-access memory (RAM) model, which allows for fast retrieval of the information stored in the tree data structure. In particular, to use this algorithm on a quantum computer, we work in the quantum RAM (QRAM) model, in which we can perform a memory access gate
	\begin{align}\label{eqn:qram}
	  \ket{j}\ket{y}\ket{x_1,\ldots,x_m} \mapsto \ket{j}\ket{y \oplus x_j}\ket{x_1,\ldots,x_m}
	\end{align}
	at unit cost. Note that implementing this operation with elementary two-qubit gates requires $\Omega(m)$ overhead \cite[Theorem 4]{Bea13}, so the cost of the algorithm described below would be larger by a factor of $\eta$ in the standard quantum circuit model, as its tree data structure occupies $m=\Theta(\eta)$ qubits.
    We leave it as an open question whether similar performance can be achieved without QRAM.
	
	\begin{lemma}\label{lem:fast-multipole}
	    Under the setting of \thm{many-rescaled-interaction-complexity}, the modified Coulomb potential evaluation oracle $U_f$ defined in \lem{U-tilde-B} can be approximately implemented
	    using $O(\eta(4 d)^{d/2}\log^3(d/\Delta))$ one- and two-qubit gates and $O(\eta(4 d)^{d/2}\log(1/\Delta))$ QRAM operations.
	    The error of this evaluation is of the same order as the error due to the difference between the modified Coulomb potential and the actual Coulomb potential.
	\end{lemma}

    \begin{proof}
We implement $U_f$ using a straightforward generalization of the multipole-based Barnes-Hut algorithm to the $d$-dimensional case. Specifically, we divide the simulation region into hypercubic cells in a tree structure, so that the out-degree of each node is $2^d$. The subdivision stops when each cell contains at most one particle or its length is at most $\delta:=\Delta^2/(D_{\max}\sqrt{d})$, where $D_{\max}$ is the maximum possible distance of two particles in the simulation region (we call such a cell a \emph{leaf cell}). Since our simulation region is a $d$-dimensional unit hypercube, we have $D_{\max}\le\sqrt{d}$.

If multiple particles are in the same leaf cell, we move them to the center of the cell (a distance of at most $\delta\sqrt{d}$) and calculate their Coulomb interactions. We now discuss the relative error caused by this movement. For two particles with charges $q_1,q_2$ and distance $D$, the error of this step is at most
        \begin{align}
            \frac{q_1q_2}{\sqrt{D^2+\Delta^2}}
            -\frac{q_1q_2}{\sqrt{(D+\delta\sqrt{d})^2+\Delta^2}}\leq\frac{q_1q_2}{\sqrt{D^2+\Delta^2}}-\frac{q_1q_2}{\sqrt{D^2+\Delta^2+2D\delta\sqrt{d}}}.
        \end{align}
        On the other hand, the error between the actual Coulomb potential and the modified Coulomb potential defined in \eqn{modified-Coulomb} is
        \begin{align}
            \frac{q_1q_2}{D}-\frac{q_1q_2}{\sqrt{D^2+\Delta^2}}.
        \end{align}
        Since our simulation region is a unit hypercube, we have $D\leq\sqrt{d}$, so the error introduced by the truncation is at most
        \begin{align}
            \frac{q_1q_2}{\sqrt{D^2+\Delta^2}}-\frac{q_1q_2}{\sqrt{D^2+3\Delta^2}},
        \end{align}
        which is of the same order as the error between the modified Coulomb potential and the actual Coulomb potential.

During the process of subdividing the simulation region, we construct and store the corresponding tree structure in the following way. For each node in the tree, we perform the preprocessing step of \lem{BH-tree} (generalized to $d$ dimensions), store the geometric center of all the electrons in the corresponding cell, and maintain $2^d$ pointers storing the memory locations of its children (for use with subsequent QRAM operations). Since the height of the tree is at most $\log(1/\Delta)$, the tree structure can be constructed with gate complexity $O(2^d\eta\log(1/\Delta))$. We assume that our QRAM operates on the entire memory space in which the tree data structure is computed, so we do not require separate QRAM writing operations during the construction of the tree.

As in \lem{BH-tree}, we can use this tree structure to approximate the Coulomb potential on any particle $P$ to any desired accuracy $\epsilon$ in a recursive way. We maintain a set of nodes $\mathcal{S}$ during the recursion, where all the cells corresponding to the nodes in $\mathcal{S}$ have the same size. Initially, we set $\mathcal{S}$ to contain only the root cell. At each recursive step, we create a set $\mathcal{S}'$ to store the children of all the nodes in $\mathcal{S}$ and set $\mathcal{S}$ to $\emptyset$. Then we enumerate through all the nodes in $\mathcal{S}'$. Let $l$ denote the length of the cell currently being processed and $D$ the distance from the geometric center of the cell to $P$. If we reach a leaf cell, we include the Coulomb interaction between the cluster of particles in this cell and $P$ by computing exact pairwise interactions. If we reach a cell with $l/D<\frac{1}{2\sqrt{d}}$, we include the Coulomb interaction between the cluster of particles in this cell and $P$ by using the $p$th-order fast multipole approximation. Since in the latter case all the particles of this cell are within a hyperball with radius $l\sqrt{d}/2$, the error of the $p$th-order fast multipole method is at most
        \begin{align}
            \frac{Q}{D-l\sqrt{d}/2}\Big(\frac{l\sqrt{d}}{2D}\Big)^{p+1},
        \end{align}
        where $Q$ denotes the absolute sum of all the charges in this cell.
        Choosing $p=\log(d/\Delta^2)=O(\log(d/\Delta))$, the cumulative error of the fast-multipole approximations is $O(Q\Delta^2/Dd)$. This is of the same order as the error of the modified Coulomb potential, which is at least of order
        \begin{align}
            \Omega\Big(\frac{Q}{D}-\frac{Q}{\sqrt{D^2+\Delta^2}}\Big)
            =\Omega\Big(\frac{Q\Delta^2}{Dd}\Big)
        \end{align}
        (since $D\leq\sqrt{d}$ for a $d$-dimensional hypercubic simulation region).
        If we did not reach a leaf cell and $l/D\geq\frac{1}{2\sqrt{d}}$, we add the corresponding node into the set $\mathcal{S}$. The recursive process stops when both $\mathcal{S}$ and $\mathcal{S}'$ are empty.

        Observe that in any recursive step, the set $\mathcal{S}$ essentially stores all the cells of some specific size $l$ that are either not leaf cells or that have distances $D$ to $P$ that are not far enough (in particular, those with $D\leq\frac{1}{2l\sqrt{d}}$). Hence, during the recursive process, the size of $\mathcal{S}$ can be upper bounded by $O(d^{d/2})$, whereas the size of $\mathcal{S}'$ can be upper bounded by $O((4d)^{d/2})$. Thus, there are in total $O((4d)^{d/2})$ lookup operations to the tree structure, or equivalently, $O((4d)^{d/2})$ queries to the QRAM.

        Moreover, since the height of the tree is at most $O(\log(1/\Delta))$, the number of recursive steps is at most $O(\log(1/\Delta))$. Thus the gate complexity for evaluating the Coulomb potential for one particle is
        \begin{align}
            O((4d)^{d/2} p^2 \log(1/\Delta))=O((4d)^{d/2}\log^3(d/\Delta)).
        \end{align}
        Therefore, the overall complexity is $O(\eta (4 d)^{d/2}\log^3(d/\Delta))$, and the overall number of QRAM queries is $O(\eta (4 d)^{d/2}\log(1/\Delta))$, as claimed.
            \end{proof}

As discussed above, the evaluation oracle $U_f$ for the modified Coulomb potential can be implemented via two different approaches. First, by directly evaluating all $\eta(\eta-1)/2$ pairwise interactions, $U_f$ can be implemented with $O(\eta^2)$ one- and two- qubit gates. Alternatively, using \lem{fast-multipole}, the cost of implementing $U_f$ by the fast multipole method is $O(\eta (4 d)^{d/2}\log^3(d/\Delta))$ one- and two- qubit gates along with $O(\eta (4 d)^{d/2}\log(1/\Delta))$ QRAM operations. Consequently, we have the following.\footnote{Appendix K of \cite{su2021fault} develops an approximate block-encoding of the modified Coulomb potential that avoids directly evaluating the potential function. That method also provides a simulation with cubic dependence on $\eta$ by using LCU techniques~\cite{childs2012hamiltonian}. Here we focus on an approach that evaluates the modified Coulomb potential explicitly.}

    \begin{corollary}\label{cor:interaction-coulomb}
    Consider the setting of \thm{many-rescaled-interaction-complexity} where $f(\vect{x})$ is the modified Coulomb potential defined in \eqn{modified-Coulomb} for particles of fixed charge (i.e., $q_i$ is independent of $i$). Then the $\eta$-particle Hamiltonian can be simulated for time $T$ with accuracy $\epsilon$ with either of the following costs:
\begin{enumerate}
\item (Direct evaluation)
$\eta^{3}(d+\eta)T\poly(\log(\eta dTg'/(\Delta\epsilon)))/\Delta$ one- and two-qubit gates, or
        \item (Fast multipole method)
        $\eta^{3}(4d)^{d/2}T\poly(\log(\eta dTg'/(\Delta\epsilon)))/\Delta$ one- and two-qubit gates and QRAM operations, if $\Delta$ is chosen small enough that the intrinsic simulation error due to the difference between the actual Coulomb potential and the modified Coulomb potential is $O(\epsilon)$,
    \end{enumerate}
where $g':=\max_t \|\Phi^{(n/2)}(\cdot,t)\|_{L^1}$ as defined in \eqn{gprime}.
    \end{corollary}

    \begin{proof}
    By \thm{many-rescaled-interaction-complexity}, the Hamiltonian can be simulated with
    \begin{align}\label{eqn:interaction-coulomb-3}
    O\bigg(\frac{T\eta^2 q^2}{\Delta}\Big(\poly(z)+\log^{2.5}(T\eta^2 q^2/(\Delta\epsilon))+\eta d\log \Big(\frac{\log(g'/\epsilon)}{\log\log(g'/\epsilon)}\Big)\Big)\frac{\log(T\eta^2 q^2/(\Delta\epsilon))}{\log\log(T\eta^2 q^2/(\Delta\epsilon))}\bigg)
    \end{align}
    one- and two- qubit gates and
    \begin{align}
        O\Big(\frac{T\eta^2 q^2}{\Delta}\cdot\frac{\log(T\eta^2 q^2/(\Delta\epsilon))}{\log\log(T\eta^2 q^2/(\Delta\epsilon))}\Big)
    \end{align}
    queries to the evaluation oracle $U_f$. The oracle can be implemented via the direct pairwise interaction with gate complexity $O(\eta^2)$, or the fast multipole method as described in \lem{fast-multipole} while introducing an error at most $O(\epsilon)$. Then the gate complexity and number of QRAM queries are as claimed.
    \end{proof}


\section{Applications}\label{sec:applications}
In this section, we study the applications of quantum simulation in real space. Our targets are several computational problems of fundamental importance in quantum chemistry, solid-state physics, and optimization.

\subsection{Quantum chemistry}\label{sec:qchem}

One of the most well-studied applications of quantum simulation is the electronic structure problem in quantum chemistry, which aims to determine ground states or low-lying excited states of the electronic Hamiltonian of molecules \cite{bauer2020quantum,cao2019quantum}. To prepare eigenstates, quantum simulation can either be used as a subroutine in quantum phase estimation, or directly used in an adiabatic state preparation procedure. Thus our improved simulations could potentially be used to give faster algorithms for these eigenstate determination problems.

However, we focus here on quantum simulation of chemical dynamics, which go beyond static properties to consider dynamical effects in chemical recation processes. To describe chemical dynamics, we start from the exact molecular Schr{\"o}dinger equation
\begin{align}\label{eqn:molecular-Schrodinger}
	i\frac{\partial}{\partial t}\Phi(\vect{x},t)=\hat{H}_m\Phi(\vect{x},t)=\Big[\hat{T}_n + \hat{H}_e\Big]\Phi(\vect{x},t),
\end{align}
where the molecular Hamiltonian $\hat{H}_m=\hat{T}_n + \hat{H}_e$ is a sum of the nuclear kinetic energy
\begin{align}\label{eqn:Born-nuclear}
	\hat{T}_n = -\sum_{A=1}^{\eta_n} \frac{1}{2M_A}\nabla^2_A
\end{align}
and the electronic Hamiltonian
\begin{align}\label{eqn:Born-electronic}
	\hat{H}_e = -\frac{1}{2}\sum_{i=1}^{\eta_e} \nabla^2_i + \frac{1}{2} \sum_{i\neq j} \frac{1}{|\hat{r}_i - \hat{r}_j|} - \frac{1}{2} \sum_{i, A} \frac{Z_A}{|\hat{r}_i - \hat{R}_A|} + \frac{1}{2} \sum_{A\neq B} \frac{Z_AZ_B}{|\hat{R}_A - \hat{R}_B|},
\end{align}
where $\hat{r}_i$ and $\hat{R}_A$ are the positions of the $i$th electron and the $A$th nucleus, and $M$ and $Z$ are the mass and the atomic number of the nucleus, respectively. This problem can be simplified using the well-known Born-Oppenheimer approximation. In this approximation, the nuclear kinetic energy $\hat{T}_n$ is neglected, the positions of the nuclei $\hat{R}$ are fixed, and the time-independent electronic Schr{\"o}dinger equation
\begin{align}\label{eqn:Born-electronic-energy}
	\hat{H}_e \chi(\hat{r},\hat{R}) = E_e(\hat{R})\chi(\hat{r},\hat{R})
\end{align}
is solved to obtain the electronic wave function $\chi(\hat{r},\hat{R})$ and the electronic energy eigenvalue $E_e(\hat{R})$. Varying $\hat{R}$ and repeatedly solving \eqn{Born-electronic-energy}, one can obtain the potential energy surface $E_e(\hat{R})$ as a function of the nuclear positions. Then, the time-dependent Schr{\"o}dinger equation of the nuclei dynamics
\begin{align}\label{eqn:Born-molecular-dynamics}
	i\frac{\partial}{\partial t}\Psi(\hat{R},t)=\Big[\hat{T}_n+E_e(\hat{R})\Big]\Psi(\hat{R},t).
\end{align}
can be solved separately.
Since the nuclei move much more slowly than the electrons, this approach often provides a good approximation and leads to a more practical method.

In classical computational chemistry, the Born-Oppenheimer approximation is often used to simplify calculations for chemical reactions, because the overall cost of calculating $E_e(\hat{R})$ from \eqn{Born-electronic-energy} with varying $\hat{R}$ and then calculating $\Psi(\hat{R},t)$ from \eqn{Born-molecular-dynamics} is less than the cost of simulating the full dynamics \eqn{molecular-Schrodinger}. Such an approximation can be trusted if the potential energy surfaces of the electronic states are well separated. However, Kassal et al.~\cite{kassal2008polynomial} pointed out that simulating the full dynamics on a quantum computer will not only yield more accurate results, but can also be faster than the Born-Oppenheimer approximation. In particular, Fig.~3 of~\cite{kassal2008polynomial} shows that the computational resources for fitting a potential energy surface $E_e(\hat{R})$ from interpolation increases exponentially with the atomic number $Z_A$, while the cost of simulating the full dynamics only increases polynomially. For chemical reactions with more than 4 atoms, it is more efficient for a quantum computer to evolve all the nuclei and electrons than to use the Born-Oppenheimer approximation~\cite{kassal2008polynomial}.

As an application of multi-particle Schr{\"o}dinger equation, we apply \thm{time-independent-many-Schrodinger-simulation} on the full dynamics \eqn{molecular-Schrodinger}. Similar to \eqn{modified-Coulomb}, generalizations of electron-electron, electron-nucleus, and nucleus-nucleus Coulomb potentials are modified as
\begin{align}
f_{\text{e-e}} &= \sum_{A\neq B} \frac{1}{\sqrt{\sum_{k=1}^3\bigl(x_{3(i-1)+k}-x_{3(j-1)+k}\bigr)^2+\Delta^{2}}}, \label{eqn:modified-molecular-Coulomb-ee} \\
f_{\text{n-e}} &= \sum_{i, A} \frac{Z}{\sqrt{\sum_{k=1}^3\bigl(x_{3(i-1)+k}-x_{3(A-1)+k}\bigr)^2+\Delta^{2}}}, \label{eqn:modified-molecular-Coulomb-en} \\
f_{\text{n-n}} &= \sum_{A\neq B} \frac{Z^2}{\sqrt{\sum_{k=1}^3\bigl(x_{3(A-1)+k}-x_{3(B-1)+k}\bigr)^2+\Delta^{2}}}. \label{eqn:modified-molecular-Coulomb-nn}
\end{align}
The total Coulomb potential $f(\vect{x})$ of the multi-particle Schr{\"o}dinger equation \eqn{many-Schrodinger-dep} is bounded by
\begin{align}
	\|f(\vect{x})\|_{L^{\infty}} \le \frac{(\eta_e+\eta_n)(\eta_e+\eta_n-1)Z^2}{2\Delta}.
\end{align}
Therefore, we find the following result.

\begin{corollary}\label{cor:molecular-dynamics-simulation}
		Consider an instance of the molecular Schr{\"o}dinger equation \eqn{molecular-Schrodinger}, with the Coulomb potentials \eqn{modified-molecular-Coulomb-ee}, \eqn{modified-molecular-Coulomb-en}, and \eqn{modified-molecular-Coulomb-nn}, and a given $T>0$. Let $g' = \max_t\|\Phi^{(n/2)}(\cdot,t)\|_{L^1}$ as in \eqn{gprime}. There exists a quantum algorithm producing a state that approximates $\Phi(\vect{x},T)$ at the nodes $\{\vect{\chi}_{\vect{l}}\}$ defined in \eqn{interpolation_nodes}, with $\ell_2$ error at most $\epsilon+O(\epsilon_0)$, with
		\begin{align}
            (\eta_e+M\eta_n)^3 TZ^2/(M\Delta)\cdot\poly(\log((\eta_e+M\eta_n)Tg'/(\Delta\epsilon))) \label{eqn:molecular-dynamics-simulation-Coulomb}
		\end{align}
		one- and two-qubit gates, along with the same number (up to poly-logarithmic factors) of QRAM operations
	\end{corollary}
	\begin{proof}
We first rescale \eqn{molecular-Schrodinger} with $\overline t := t/M_A$ and $\overline T := T/M_A$. Now the molecular dynamics becomes
\begin{align}
	i\frac{\partial}{\partial \overline t}\Phi(\vect{x},\overline t)=M_A\hat{H}_m\Phi(\vect{x},\overline t)=M_A\Big[\hat{T}_n + \hat{H}_e\Big]\Phi(\vect{x},\overline t).
\end{align}
This means we ``accelerate the time'' to capture the movement of nuclei, whose nuclear kinetic energy is rescaled as $M\hat{T}_n = -\sum_{A=1}^{\eta_n} \frac{1}{2}\nabla^2_A$. Then we treat the rescaled electronic Hamiltonian $M\hat{H}_e$ as a Hamiltonian of $M\eta_e$ of electrons, with a rescaled potential bounded by $M\frac{(\eta_e+\eta_n)(\eta_e+\eta_n-1)Z^2}{2\Delta}$. We then find the gate complexity \eqn{molecular-dynamics-simulation-Coulomb} by straightforwardly applying \cor{interaction-coulomb} with $d=3$.
\end{proof}

Compared to the previous work for simulating the full dynamics of electrons and nuclei on a quantum computer, Kassal et al.~\cite{kassal2008polynomial} represent real-space quantum dynamics using a discrete system of qubits and apply product formulas to propagate the system. The query complexity of this approach scales quadratically with the particle number, and the gate complexity should be larger by an additional factor of the particle number, due to the analysis of \thm{many-rescaled-interaction-complexity}. The query and gate complexity of our approach matches the dependence of the particle number $\eta_e$ and $\eta_n$. Furthermore, our analysis explicitly bounds the complexity as a function of $T$ and $d$, and achieves $\poly(\log(1/\epsilon))$ dependence on the error $\epsilon$.


\subsection{Uniform electron gas}

The uniform electron gas, also known as jellium, is a simple theoretical model of delocalized electrons in a metal. It considers a large system of interacting electrons with a homogeneous jelly-like continuum of positive background charge, so that the entire system is charge-neutral \cite{giuliani2005quantum}. Jellium is considered a good approximation of electrons confined in semiconductor wells, or valence electron distributions of alkali metals such as sodium. Despite its simplicity, the jellium model can be classically hard to simulate in some regimes, and it is widely used as a benchmark problem for new classical quantum simulation methods \cite{babbush2017low}.

For concreteness, we consider $\eta$ electrons in a $3$-dimensional cubic box $\Omega = [0,1]^{3\eta}$ and assume the wave function can be treated as periodic on this domain. As in the discussion in \sec{multi-particles}, we can realize this assumption by embedding the system into a sufficiently large periodic box. Let $m$ denote the electron mass and let $e$ denote the unit charge. Then the jellium Hamiltonian reads
\begin{align}\label{eqn:jellium-hamiltonian}
	\hat{H} = \sum^\eta_{i=1} \frac{\hat{p}^2_i}{2m} + \frac{1}{2} \sum_{i\neq j} \frac{e^2}{|\hat{r}_i - \hat{r}_j|} + \hat{H}_{e-b} + \hat{H}_{b-b},
\end{align}
where $\hat{p}_i$, $\hat{r}_i$ are the momentum and position operators for the $i$th electron, and $\hat{H}_{e-b}$ and $\hat{H}_{b-b}$ are the electron-background and background-background interactions, respectively. These background terms take the form
\begin{align}
	\hat{H}_{e-b} &= -e^2 \eta \sum^\eta_{i = 1}\int_\Omega \frac{1}{|\hat{r}_i - R|}  \d R, \\
	\hat{H}_{b-b} &= e^2 \eta^2 \int_{\Omega^2} \frac{1}{|R - R'|} \d R \d R'.
\end{align}

The background-background interaction operator $\hat{H}_{b-b}$ is constant. Because of the homogeneity of jellium, the electron-background operator $\hat{H}_{e-b}$ is also constant in the periodic cubic box $\Omega$ \cite{babbush2017low}. For simplicity, we assume $\hat{H}_{b-b} = \hat{H}_{e-b} = 0$ as constant operators merely add a global phase to the wave function in the quantum evolution.

Similarly to Eq.~\eqn{modified-Coulomb}, we handle the singularity in the electron-electron interaction by introducing a modified Coulomb potential
\begin{align}\label{eq:jellium-coulomb}
    f_{e-e}= \frac{1}{2}\sum_{i\neq j} \frac{e^2}{\sqrt{\sum_{k=1}^3\bigl(x_{3(i-1)+k}-x_{3(j-1)+k}\bigr)^2+\Delta^{2}}}.
\end{align}
This potential satisfies
\begin{align}\label{eqn:ee-upperbound}
    \|f_{e-e}(\vect{x})\|_{L^{\infty}}\leq\frac{\eta(\eta+1)e^2}{2\Delta}.
\end{align}
Therefore, we obtain the following result.
\begin{corollary}\label{cor:jellium-simulation}
Consider an instance of the Schr{\"o}dinger equation \eqn{jellium-hamiltonian} describing a uniform gas of $\eta$ electrons in $3$-dimensional space, with the modified Coulomb potential between the electrons defined in \eq{jellium-coulomb}. Then the dynamics for time $T$ can be simulated with accuracy $\epsilon+O(\epsilon_0)$ using
\begin{align}
    \eta^3T\poly(\log(\eta Tg'/(\Delta\epsilon)))/\Delta
\end{align}
one- and two-qubit gates and QRAM operations, where $g'=\max_t \|\Phi^{(n/2)}(\cdot,t)\|_{L^1}$ as defined in \eqn{gprime}, and $\epsilon_0$ denotes the intrinsic simulation error due to the difference between the actual Coulomb potential and the modified Coulomb potential.
\end{corollary}
\begin{proof}
It suffices to replace $q$ by $e$ in the proof of \cor{interaction-coulomb} and set $d=3$ to obtain the claimed query and gate complexities.
\end{proof}

\subsection{Optimization}

In recent work on quantum algorithms for nonconvex optimization \cite{zhang2020quantum}, Schr{\"o}dinger equation simulation was used as a technique for escaping from saddle points, a key challenge in optimization theory. Ref.~\cite{zhang2020quantum} demonstrated that for a $d$-dimensional nonconvex function $f$, a quantum speedup can be achieved for the problem of escaping from saddle point if we replace the uniform perturbation step in the classical state-of-the-art algorithm \cite[Algorithm 2]{jin2018accelerated} by a perturbation obtained by simulating the Schr{\"o}dinger equation
\begin{align}
    i\frac{\partial}{\partial t}\Phi=\Big[-\frac{r_0^2}{2}\nabla^2+\frac{1}{r_0^2}f(\vect{x})\Big]\Phi
    \label{eqn:Schrodinger-optimization}
\end{align}
in a hyperball region with a small radius $r_0$. The resulting quantum algorithm \cite[Algorithm 2]{zhang2020quantum} is presented as \algo{PGD+QS}, with the QuantumSimulation subroutine \cite[Algorithm 1]{zhang2020quantum} in \algo{PGD+QS} shown in \algo{Optimization-Simulation}.

\begin{algorithm}[htbp]
\caption{Perturbed Gradient Descent with Quantum Simulation.}
\label{algo:PGD+QS}
\For{$t=0,1,...,T$}{
\If{$\|\nabla f(\vect{x}_{t})\|\leq\epsilon$}{
$\xi\sim$QuantumSimulation($\vect{x}_t$,$r_0$,$\mathscr{T}'$,$f(\vect{x})-\left\<\nabla f(\vect{x}_t),\vect{x}-\vect{x}_t\right\>$)\label{lin:QuantumSimulation}\;
$\Delta_t\leftarrow\frac{2\xi}{3\|\xi\|}\sqrt{\frac{\rho}{\epsilon}}$\;
$\vect{x}_{t}\leftarrow\mathop{\arg\min}_{\zeta\in\left\{\vect{x}_t+\Delta_t,\vect{x}_t-\Delta_t\right\}}f(\zeta)$\;
}
$\vect{x}_{t+1}\leftarrow\vect{x}_{t}-\eta\nabla f(\vect{x}_{t})$\;
}
\end{algorithm}

\begin{algorithm}[htbp]
\caption{QuantumSimulation($\tilde{\vect{x}},r_{0},t,f(\cdot)$).}
\label{algo:Optimization-Simulation}
Put a Gaussian wave packet into the potential field $f$, with the initial state
\begin{align}\label{eqn:ground_state_Phi0}
\Phi_{0}(\vect{x})=\Big(\frac{1}{2\pi}\Big)^{d/4}\frac{1}{r_{0}^{d/2}}\exp(-(\vect{x}-\tilde{\vect{x}})^{2}/4r_{0}^{2});
\end{align}
Simulate its evolution in potential field $f$ with the Schr\"odinger equation for time $t$\;
Measure the position of the wave packet and output the measurement outcome.
\end{algorithm}

Using \prop{rescaled-interaction-complexity}, we demonstrate that our real-space simulation algorithm can be used to perform \algo{Optimization-Simulation} and thereby obtain a better complexity for \algo{PGD+QS}.

\begin{lemma}\label{lem:optimization-simulation}
Suppose $f(\vect{x})\colon \R^d \to \R$ is continuously differentiable and have a saddle point at $\vect{x}=0$ with $f(\vect{0}) = 0$. Suppose $f(\vect{x}) = \frac{1}{2}\vect{x}^T \mathcal{H} \vect{x}$ in a hypercubic domain $\Omega = \{\vect{x}\in \R^d\colon \|\vect{x}\|\le M\}$ for some universal upper bound $M > 0$. Consider the Schr\"odinger equation \eqn{Schrodinger-optimization} defined on the compact domain $\Omega$ with Dirichlet boundary conditions.\footnote{As in simulations of multi-particle quantum dynamics, we consider a wave function defined on an underlying periodic $d$-dimensional hypercubic domain. However, as before, this technique is capable of handling the non-periodic problem at hand. Specifically, we slightly ``mollify'' $f$ near the boundary of the domain to respect periodicity. This mollification does not have a significant impact for optimization because our simulation time is short and the wave function has little chance to hit the boundary.}
Given the quantum evaluation oracle $U_{f}|\x\>|0\>=|\x\>|f(\x)\>$ encoding the potential function $f$ and an arbitrary initial state at time $0$, the evolution of \eqn{Schrodinger-optimization} for time $t>0$ can be simulated with precision $\epsilon$ using
\begin{align}
    O\bigg(\frac{M^2\|\mathcal{H}\|t}{r_0^2}\cdot\frac{\log(M^2\|\mathcal{H}\|t/(r_0^2\epsilon))}{\log\log(M^2\|\mathcal{H}\|t/(r_0^2\epsilon))}\bigg)=\tilde{O}(t\log(t/\epsilon))
\end{align}
queries to $U_{f}$, with
\begin{align}
    &O\bigg(\frac{M^2\|\mathcal{H}\|t}{r_0^2}\Big(\poly(z)+\log^{2.5}\Bigl(\frac{M^2\|\mathcal{H}\|t}{r_0^2\epsilon}\Bigr)
    +d\log \Big(\frac{\log(g'/\epsilon)}{\log\log(g'/\epsilon)}\Big)\Big)\frac{\log(M^2\|\mathcal{H}\|t/(r_0^2\epsilon))}{\log\log(M^2\|\mathcal{H}\|t/\epsilon)}\bigg)
\end{align}
additional one- and two-qubit gates. Here $g'=\max_t \|\Phi^{(n/2)}(\cdot,t)\|_{L^1}$ as defined in \eqn{gprime}.
\end{lemma}

\begin{proof}
    We rescale \eqn{Schrodinger-optimization} as
    \begin{align}\label{eqn:rescaled-Schrodinger-optimization}
        i\frac{\partial}{\partial t}\Psi=\Big[-\frac{1}{2}\nabla^{2}+\frac{1}{r_{0}^2}f(r_0\cdot\vect{x})\Big]\Psi
    \end{align}
    on the compact domain $\Omega'=\{\vect{x}\in \R^d: \|\vect{x}\| \le M/r_0\}$ with Dirichlet boundary conditions. Observe that the wave function $\Psi$ in \eqn{rescaled-Schrodinger-optimization} and the original wave function $\Phi$ in \eqn{Schrodinger-optimization} are related as
    \begin{align}
        \Phi(\vect{x},t)=\Psi(\vect{x}/r_0,t),
    \end{align}
    so it suffices to simulate \eqn{rescaled-Schrodinger-optimization} instead of \eqn{Schrodinger-optimization}, where the potential function is $f'(\vect{x})=f(r_0\cdot\vect{x})/r_0^2$ satisfying
    \begin{align}
        |f'(\vect{x})|\leq M^2\|\mathcal{H}\|/r_0^2,\qquad\forall \vect{x}\in\Omega'.
    \end{align}
    Furthermore, note that the value of $f'$ is time-independent, indicating the oracles $O_{\nrm}$ and $O_{\inv}$ act trivially. Thus we can determine the query complexity to $U_f$ from \prop{rescaled-interaction-complexity}, giving
    \begin{align}
        O\bigg(\frac{M^2\|\mathcal{H}\|t}{r_0^2}\cdot\frac{\log(M^2\|\mathcal{H}\|t/(r_0^2\epsilon))}{\log\log(M^2\|\mathcal{H}\|t/(r_0^2\epsilon))}\bigg).
    \end{align}
    Absorbing all absolute constants in the big-$O$ notation, the query complexity is $\tilde{O}(t\log(t/\epsilon))$.

    Finally, the gate complexity can also be derived using \prop{rescaled-interaction-complexity}, giving
    \begin{align}
\hspace{-2mm}O\bigg(\frac{M^2\|\mathcal{H}\|t}{r_0^2}\Big(\poly(z)+\log^{2.5}\Big(\frac{M^2\|\mathcal{H}\|t}{r_0^2\epsilon}\Big)+d\log \Big(\frac{\log(g'/\epsilon)}{\log\log(g'/\epsilon)}\Big)\Big)\frac{\log(M^2\|\mathcal{H}\|t/(r_0^2\epsilon))}{\log\log(M^2\|\mathcal{H}\|t/(r_0^2\epsilon))}\bigg)
\end{align}
as claimed.
\end{proof}

Note that under the parameter choice of \cite{zhang2020quantum}, there exists a small enough constant $C_r$ satisfying
\begin{align}
    r_0=C_r\cdot M.
\end{align}

The following lemma characterizes the number of calls to the quantum simulation procedure in the entire quantum optimization algorithm.

\begin{lemma}[{\cite[Theorem 3]{zhang2020quantum}}]\label{lem:simulation-number}
Let $f\colon \R^d \to \R$ be a twice-differentiable function satisfying
\begin{align}
    \|\nabla f(\vect{x}_1)-\nabla f(\vect{x}_2)\|\leq\ell\|\vect{x}_1-\vect{x}_2\|\quad\forall\vect{x}_1,\vect{x}_2\in\mathbb{R}^d
\end{align}
and
\begin{align}
    \|\mathcal{H}(\vect{x}_1)-\mathcal{H}(\vect{x}_2)\|\leq\rho\|\vect{x}_1-\vect{x}_2\|
    \quad\forall\vect{x}_1,\vect{x}_2\in\mathbb{R}^d
\end{align}
where $\mathcal{H}$ is the Hessian matrix of $f$ and $\rho$, $\ell$ are constants. Then for any $\epsilon>0$, \algo{PGD+QS} finds an $\epsilon$-approximate local minimum with success probability at least $2/3$ using $O\big(\frac{f(\vect{x}_0)-f^*}{\epsilon^{1.5}}\big)$ calls to the quantum simulation of Eq.~\eqn{Schrodinger-optimization}, where in each call Eq.~\eqn{Schrodinger-optimization} is simulated for time
\begin{align}
    \mathscr{T}':=\frac{8}{(\rho\epsilon)^{1/4}}\log\Big(\frac{\ell(f(\vect{x}_0)-f^*)}{\epsilon^2\sqrt{\rho}}\Big(d+2\log(3(f(\vect{x}_0)-f^*)/\epsilon^{1.5})\Big)\Big),
\end{align}
where $\vect{x}_0$ is the initial point of the algorithm and $f^*$ is the global minimum of $f$.
\end{lemma}

We can analyze the overall cost of \algo{PGD+QS} by combining \lem{simulation-number} with a bound on the cost of quantum simulation.

\begin{corollary}\label{cor:esc-saddle-query}
Let $f\colon \R^d \to \R$ be a twice-differentiable function satisfying
\begin{align}
    \|\nabla f(\vect{x}_1)-\nabla f(\vect{x}_2)\|\leq\ell\|\vect{x}_1-\vect{x}_2\|\quad\forall\vect{x}_1,\vect{x}_2\in\mathbb{R}^d,
\end{align}
and
\begin{align}
    \|\mathcal{H}(\vect{x}_1)-\mathcal{H}(\vect{x}_2)\|\leq\rho\|\vect{x}_1-\vect{x}_2\|
    \quad\forall\vect{x}_1,\vect{x}_2\in\mathbb{R}^d,
\end{align}
where $\mathcal{H}$ is the Hessian matrix of $f$ and $\rho$, $\ell$ are constants,. Then for any $\epsilon>0$, \algo{PGD+QS} outputs an $\epsilon$-approximate local minimum with success probability at least $2/3$ using $\tilde{O}\big(\frac{f(\vect{x}_0)-f^{*}
}{\epsilon^{1.75}} \log d\big)$ queries to the evaluation oracle $U_f$, where $\vect{x}_0$ is the initial point of the algorithm, $f^*$ is the global minimum of $f$, and the $\tilde{O}$ notation omits poly-logarithmic factors as in \fnote{tilde-O}.
\end{corollary}

\begin{proof}
In \algo{PGD+QS}, queries to $U_f$ are only performed in the quantum simulation step in \lin{QuantumSimulation} of \algo{PGD+QS}. By \lem{optimization-simulation} and \lem{simulation-number}, each quantum simulation call uses
\begin{align}
    O\Big(\|\mathcal{H}\|\mathscr{T}' \frac{\log(\|\mathcal{H}\|\mathscr{T}'/\epsilon)}{\log\log(\|\mathcal{H}\|\mathscr{T}'/\epsilon)}\Big)
\end{align}
queries to $U_f$.
By the $\ell$-smoothness of $f$, we have $\|\mathcal{H}\|\leq\ell$. Therefore
\begin{align}
     O\Big(\|\mathcal{H}\|\mathscr{T}' \frac{\log(\|\mathcal{H}\|\mathscr{T}'/\epsilon)}{\log\log(\|\mathcal{H}\|\mathscr{T}'/\epsilon)}\Big)
     =\tilde{O}\Big(\frac{\log d}{\epsilon^{1/4}}\Big).
\end{align}
Since there are in total $O\big(\frac{f(\vect{x}_0)-f^{*}}{\epsilon^{1.5}}\big)$ quantum simulation calls, the overall query complexity is
\begin{align}
    O\Big(\frac{f(\vect{x}_0)-f^{*}}{\epsilon^{1.5}}\Big) \tilde{O}\Big(\frac{\log d}{\epsilon^{1/4}}\Big)
    =\tilde{O}\Big(\frac{f(\vect{x}_0)-f^{*}}{\epsilon^{1.75}} \log d\Big)
\end{align}
as claimed.
\end{proof}
For comparison, the previous result~\cite{zhang2020quantum} uses $\tilde{O}\big(\frac{f(\vect{x}_0)-f^*}{\epsilon^{1.75}} \log^2 d\big)$ quantum evaluation queries, so our simulation approach achieves a quadratic speedup in terms of $\log d$. Compared to classical algorithms for escaping from saddle points, we achieve polynomial speedup over the seminal work of Jin et al.~\cite{jin2018accelerated} which makes $\tilde{O}\big(\frac{f(\vect{x}_0)-f^*}{\epsilon^{1.75}} \log^6 d\big)$ gradient queries, and match the iteration number of the state-of-the-art result~\cite{zhang2021escape} which makes $\tilde{O}\big(\frac{f(\vect{x}_0)-f^*}{\epsilon^{1.75}} \log d\big)$ classical gradient queries.

The fact that this simulation-based quantum algorithm uses only evaluation queries instead of gradient queries enables a larger range of applications than classical approaches~\cite{jin2018accelerated,zhang2021escape}, especially for problems where the gradient values are not directly available. Although in principle one can use Jordan's algorithm \cite{jordan2005fast} to replace the classical gradient queries in \cite{zhang2021escape} by quantum evaluation queries with logarithmic overhead in query complexity, Jordan's algorithm must be implemented with high precision to detect feasible directions for escaping from saddle points since the gradients near saddles have small norms. Therefore, the number of qubits required in an approach based on Jordan's algorithm may be large.


\section{Conclusions}\label{sec:conclusions}

In this paper, we conducted a systematic study of quantum algorithms for simulating real-space dynamics.
We also gave applications to several computational problems in quantum chemistry, solid-state physics, and optimization.

Our work also leaves several other natural open questions for future investigation:
\begin{itemize}
\item Can we achieve better bounds using additional assumptions about the potential? For instance, in Hamiltonian simulation, faster quantum algorithms can be achieved using commutator bounds for various classes of Hamiltonians~\cite{childs2019theory} and Lieb-Robinson bounds for geometrically local Hamiltonians~\cite{haah2018quantum}. It is natural to investigate whether such ideas can be incorporated into simulation of real-space dynamics.

\item Can we prove lower bounds on the number of queries to the quantum evaluation oracle? There is an $\Omega(T)$ quantum lower bound for simulating finite-dimensional systems known as the ``no-fast-forwarding theorem''~\cite{berry2007efficient}, and we believe that the same bound holds for real-space simulation. Can we prove lower bounds in terms of other parameters such as the spatial dimension and the simulation error?

\item Since the Fourier spectral method is widely used to solve PDEs, it might be possible to generalize our results to other kinds of PDEs with various boundary conditions. Due to the diverse conditions and properties related to the existence, uniqueness, and well-posedness of the solutions, it may be difficult to merge solvers of different types of PDEs into a unified algorithmic framework. Future work might begin by focusing on PDEs with a similar structure to the Schr{\"o}dinger equation.

\item One approach to simulating quantum field theories is to relate them to multi-particle Schr\"odinger equations \cite{zalka1998efficient}. Related ideas have been discussed in the literature \cite{jordan2012quantum, preskill2019simulating}, but to the best of our knowledge, a systematic and rigorous treatment of the simulation cost is lacking.

\item The recent paper \cite{shi2020learning} proved that the probability density function $\psi_s$ of a continuous
version of the stochastic gradient descent algorithm satisfies the Fokker–Planck–Smoluchowski equation
\begin{align}
\frac{\partial \psi_s}{\partial t}=\frac{1}{2}\nabla^2\psi_s+\Big(\frac{\nabla^2f}{2}-\frac{\|\nabla f\|_2^2}{2}\Big)\psi_s,
\end{align}
which is similar to the Schr{\"o}dinger equation except for the absence of the imaginary unit $i$. It might be worth exploring whether the dynamics of stochastic gradient descent could be accelerated using real-space quantum simulation.

\item We generalize the Barnes-Hut algorithm to construct a quantum evaluation oracle for Coulomb interactions with $O(\eta(4d)^{d/2}\log^3(d/\Delta))$ gates and $O(\eta(4d)^{d/2}\log(1/\Delta))$ additional QRAM operations, achieving a quadratic gate complexity advantage in $\eta$ compared to the straightforward method of evaluating all $\eta^2$ pairwise interactions (for constant dimension $d$). Can similar simulation performance can be achieved without QRAM?
\end{itemize}


\section*{Acknowledgments}

We thank Ryan Babbush, Lin Lin, and Yuan Su for feedback on a preliminary version of this paper.

AMC received support from the National Science Foundation (grant CCF-1813814 and QLCI grant OMA-2120757) and the Department of Energy, Office of Science, Office of Advanced Scientific Computing Research, Quantum Algorithms Teams and Accelerated Research in Quantum Computing programs.
JL was supported by the National Science Foundation (grant CCF-1816695).
TL was supported by the the National Science Foundation (grant PHY-1818914), a Samsung Advanced Institute of Technology Global Research Partnership, a startup fund from Peking University, and the Advanced Institute of Information Technology, Peking University.
JPL was supported by the National Science Foundation (grant CCF-1813814) and the National Science Foundation Quantum Information Science and Engineering Network (QISE-NET) triplet award (DMR-1747426).



\appendix

	\section{Notation}\label{append:notation}
	Throughout the paper, $\N$ denotes the set of all positive integers, and $\N_{0}:=\N\cup\{0\}$. We also let $\range{n}:=\{1,\ldots,n\}$ and $\rangez{n+1}:=\{0,1,\ldots,n\}$. We introduce a variety of norms that are used in our analysis.

	\begin{definition}[Vector norms]
		For a vector $\vect{v} = (v_1,...,v_d)^T \in \C^d$, and $1 \le p \le \infty$, the vector $l_p$ norm of $\vect{v}$ is
		\begin{align}
			\|\vect{v}\|_p := \begin{cases}
				\left(\sum^d_{j=1} |v_j|^p\right)^{1/p} & 1 \le p < \infty\\
				\max_{j=1,2,...,d} |v_j| & p = \infty.
			\end{cases}
		\end{align}
	\end{definition}

	\begin{definition}[Matrix norms]
		For a matrix $A \in M^{d\times d}(\C)$, the Schatten $p$-norm of $A$ is
		\begin{align}
			\|A\|_p := \begin{cases}
				\Tr\left(\sqrt{A^\dagger A}\right) & p = 1\\
				\sqrt{\Tr(A^\dagger A)} & p = 2\\
				\max_{\|\ket{\psi}\|_2 = 1} \|A\ket{\psi}\|_2 & p = \infty.
			\end{cases}
		\end{align}
		The Schatten $\infty$-norm is also referred to as the spectral norm, as it equals the largest singular value of $A$. In this paper, we usually let $\|A\|$ denote the spectral norm of $A$ when the context is unambiguous.
	\end{definition}

	\begin{definition}[Function norms]
		If $f\colon \Omega \to \C$ is a continuous function defined on a set $\Omega \subset \R^d$, the $L^p$-norm of the function $f$ is
		\begin{align}
			\|f\|_p := \begin{cases}
				\left(\int_\Omega |f(x)|^p\,\d x\right)^{1/p} & 1 \le p < \infty\\
				\sup_{x \in \Omega} |f(x)| & p = \infty.
			\end{cases}
		\end{align}
	\end{definition}

	We combine these notations to define norms for vector-valued (or operator-valued) functions. If $\vect{v}\colon [0,t] \to \C^d$ is a continuous vector-valued function, where the $j$th coordinate at time $s$ is $v_j(s)$, then the notation $\|\vect{v}\|_{p,q}$ indicates that we take the $l_p$ norm $\|\vect{v}(s)\|_p$ for every $s \in [0,t]$ and compute the $L^q$-norm of the resulting scalar function. For instance,
	\begin{align}
		\|\vect{v}\|_{1,\infty} := \sup_{s\in [0,t]} \sum^d_{j=1}|v_j(s)|.
	\end{align}

	\begin{definition}
		The largest matrix element of $A$ in absolute value is denoted
		\begin{align}
			\|A\|_{\max} := \max_{j,k} |A_{j,k}|.
		\end{align}
		Note that $\|A\|_{\max}$ is a vector norm of $A$ but does not satisfy the sub-multiplicative property of a matrix norm.
	\end{definition}

\begin{lemma}[{\cite[Lemma 1]{childs2009limitations}}]
	For any Hermitian matrix $A \in \C^N \times \C^N$, we have
	\begin{align}
		\|A\|_{\max} \le \|A\| \le N\|A\|_{\max}.
	\end{align}
\end{lemma}

Finally, we introduce the notion of block encoding and some related concepts.

\begin{definition}[Block encoding]\label{def:block-encoding}
	Suppose that $A$ is an $s$-qubit operator, $\alpha,\epsilon \in \R^+$, and $a\in \N$. Then we say the $(s+a)$-qubit unitary $U$ is an $(\alpha, a, \epsilon)$-block-encoding of $A$ if
	\begin{align}
		\|A - \alpha\left(\bra{0}^{\otimes a}\otimes\mathbbmss{1}_{s}) U (\ket{0}^{\otimes a}\otimes\mathbbmss{1}_{s})	\right)\| \le \epsilon.
	\end{align}
\end{definition}

The following result shows that sparse matrices can be efficiently block-encoded.

\begin{lemma}[{\cite[Lemma 48]{gilyen2019quantum}}]\label{lem:block-encoding-diagonal}
	Let $\Gamma\in\C^{2^\omega\times 2^{\omega}}$ be a matrix that is $s_r$-row-sparse and $s_c$-column-sparse, and each element of $\Gamma$ has absolute value at most $1$. Suppose that we have access to the following sparse-access oracles acting on two $\omega+1$ qubit registers:
    \begin{align}
        O_r\colon\ket{i}\ket{k}\to\ket{i}\ket{r_{ik}},]\qquad\forall i\in[2^\omega]-1,\,k\in[s_r],
    \end{align}
    and
    \begin{align}
        O_c\colon\ket{\ell}\ket{j}\to\ket{c_{\ell j}}\ket{j}\qquad\forall\ell\in[s_c],\,j\in [2^{\omega}]-1,
    \end{align}
    where $r_{ij}$ is the index of the $j$th non-zero entry in the $i$th row of $\Gamma$, or $j+2^{\omega}$ if there are fewer than $i$ non-zero entries; and similarly, $c_{ij}$ is the index of the $i$th non-zero entry in the $j$th column of $\Gamma$, or $i+2^{\omega}$ if there are fewer than $j$ non-zero entries. Additionally, assume that we have access to an oracle $O_{\Gamma}$ for the entries of $A$:
    \begin{align}
        O_{\Gamma}\colon\ket{i}\ket{j}\ket{0}^{\otimes z}\to\ket{i}\ket{j}\ket{\Gamma_{ij}},\qquad\forall i,j\in[2^{\omega}]-1,
    \end{align}
    where $\Gamma_{ij}$ is a $z$-bit binary encoding of the $ij$-matrix element of $\Gamma$. Then we can implement a $(\sqrt{s_rs_c},\omega+3,\epsilon)$-block-encoding of $\Gamma$ with a single use of $O_r$, $O_c$, two uses of $O_{\Gamma}$, and $O\big(\omega+\log^{2.5}\big(\frac{s_rs_c}{\epsilon}\big)\big)$ additional one- and two-qubit gates.
\end{lemma}

We also define a block encoding of a time-dependent Hamiltonian.

\begin{definition}[HAM-T]\label{defn:hamt}
		Let $\widetilde{H}_I(t)\colon [0,\varsigma]\rightarrow\C^{2^{w}\times2^{w}}$ be the rescaled interaction-picture Hamiltonian in \eqn{interaction-picture-hamt}. We define a block-encoding unitary oracle $\hamt\in\C^{M2^{n_{a}+w}\times M2^{n_{a}+w}}$ satisfying
		\begin{align}
			\hamt=\begin{pmatrix}
				\mathfrak{H} & \cdot\\
				\cdot &\cdot
			\end{pmatrix}
		\end{align}
		in which
		\begin{align}
			\mathfrak{H}=\diag\left[\widetilde{H}_I(t)(0),\widetilde{H}_I(\varsigma/M),\cdots,\widetilde{H}_I\left((M-1)\varsigma/M\right)\right].
		\end{align}
	\end{definition}

	A direct calculation shows that
		\begin{align}
			(\bra{0}_{a}\otimes\mathbbmss{1})\hamt(\ket{0}_{a}\otimes\mathbbmss{1})=\sum_{m=0}^{M-1}\ket{m}\bra{m}\otimes \widetilde{H}_I(m\varsigma/M).
		\end{align}

	\section{Exponential convergence of the Fourier spectral method}\label{append:spectral}

	Here we establish a concrete error estimate for the Fourier spectral method.

	\begin{lemma}[{\cite[Theorem 20]{boyd2001chebyshev}}]\label{lem:pseudospectral-error}
		Let $f(x)$ have the exact, infinite trigonometric polynomial representation
		\begin{align}
			f(x) = \frac{\alpha_0}{2} + \sum_{k=0}^{\infty}\alpha_k\cos(kx) + \sum_{k=0}^{\infty}\beta_k\sin(kx).
		\end{align}
		Let $S_{n+1}(x)$ denote the trigonometric polynomial that interpolates to a function $f(x)$ on \eqn{interpolation_nodes},
		\begin{align}
			S_{n+1}(x) = \frac{a_0}{2} + \sum_{k=0}^{n/2}a_k\cos(kx) + \sum_{k=0}^{n/2}b_k\sin(kx).
		\end{align}
		Then the error from the Fourier spectral method satisfies
		\begin{align}\label{eqn:pseudospectral-error}
			|f(x)-S_{n+1}(x)| \le 2\sum_{k=n/2+1}^{\infty} (|\alpha_k|+|\beta_k|).
		\end{align}
	\end{lemma}

	\lem{pseudospectral-error} tells us that the error from the Fourier spectral method is bounded by an infinite sum of the Fourier coefficients of $f(x)$. It is well-known that the Fourier coefficients decay rapidly in terms of $n$, which is formally established as follows.

	\begin{lemma}\label{lem:Fourier-coefficient}
		Let $f(x)$ have the exact, infinite Fourier series representation
		\begin{align}
			f(x) = \sum_{k=-\infty}^{\infty}\hat f(k)e^{ikx}.
		\end{align}
		Assume $f(x)\in\C^p$. Then the Fourier coefficient $\hat f(k)$ satisfies
		\begin{align}\label{eqn:coefficient_decay}
			|\hat f(k)| \le \frac{\|f^{(p)}\|_{L^1}}{k^p}.
		\end{align}
	\end{lemma}

	\begin{proof}
		The Fourier coefficient $\hat f(k)$ on $[-\pi,\pi]$ is
		\begin{align}
			\hat f(k)  = \frac{1}{2\pi}\int_{-\pi}^{\pi}f(x)e^{-ikx}\d{x}.
		\end{align}
		Using $f(x)\in \C^p$ and $f^{(k)}(-\pi)=f^{(k)}(\pi)$ for arbitrary $k\in\N$, integrating by parts $p$ times, we obtain
		\begin{align}
			\hat f(k) = \frac{(-i)^k}{2\pi k^p}\int_{-\pi}^{\pi}f^{(p)}(x)e^{-ikx}\d{x}.
		\end{align}
		Thus, the magnitude of the Fourier coefficient has the gbound
		\begin{align}
			|\hat f(k)| \le \frac{1}{2\pi k^p}\int_{-\pi}^{\pi}|f^{(p)}(x)|\d{x} = \frac{\|f^{(p)}\|_{L^1}}{2\pi k^p}
		\end{align}
		as claimed.
	\end{proof}

	Note that the coefficients of the trigonometric polynomial and Fourier series of $f(x)$ are related as
	\begin{align}\label{eqn:coefficient-transform}
		\hat f(k) = \begin{cases}
			\frac{1}{2}(\alpha_k + i\beta_k), \qquad & k\in\Z^+; \\
			\frac{1}{2}\alpha_0, \qquad & k=0; \\
			\frac{1}{2}(\alpha_k - i\beta_k), \qquad & k\in\Z^-.
		\end{cases}
	\end{align}

	Combining \lem{pseudospectral-error} with \lem{Fourier-coefficient} gives an error bound for the Fourier spectral method.

	\begin{lemma}\label{lem:pseudospectral-estimate}
		Under the same assumptions as \lem{pseudospectral-error}, the error from the Fourier spectral method satisfies
		\begin{align}\label{eqn:pseudospectral-estimate}
			|f(x)-S_{n+1}(x)| \le \frac{2\|f^{(p)}\|_{L^1}}{\pi(p-1)(n/2+1)^{p-1}}.
		\end{align}
	\end{lemma}

	\begin{proof}
		Using \eqn{pseudospectral-error} and \eqn{coefficient-transform}, we upper bound the error by
		\begin{align}
			|f(x)-S_{n+1}(x)| \le 2\sum_{k=n/2+1}^{\infty}(|\alpha_k|+|\beta_k|) \le 4\sum_{k=n/2+1}^{\infty}|\hat f(k)|.
		\end{align}
		Plugging in \eqn{coefficient_decay} gives
		\begin{align}
			\sum_{k=n/2+1}^{\infty}|\hat f(k)| \le \sum_{k=n/2+1}^{\infty}\frac{\|f^{(p)}\|_{L^1}}{2\pi k^p}.
		\end{align}
		Such an infinite sum can be bounded by an infinite integral, using
		\begin{align}
			\sum_{k=n/2+1}^{\infty}\frac{1}{k^p} \le \sum_{k=n/2+1}^{\infty}\int_{k-1}^k\frac{1}{y^p}\d{y} = \int_{n/2}^{\infty}\frac{1}{y^p}\d{y} = \frac{1}{(p-1)(n/2+1)^{p-1}}.
		\end{align}
		Therefore
		\begin{align}
			|f(x)-S_{n+1}(x)| \le \frac{2\|f^{(p)}\|_{L^1}}{\pi(p-1)(n/2+1)^{p-1}}
		\end{align}
		as claimed.
	\end{proof}

	Since we are concerned with analytic $f(x)$, it suffices to choose $p=n/2$ (for $n\ge6$) to obtain the following result.

	\begin{lemma}\label{lem:pseudospectral-estimate-smooth}
		If $f(x)$ is analytic, according to \lem{pseudospectral-error}, the error from the Fourier spectral method satisfies
		\begin{align}\label{eqn:pseudospectral-estimate-smooth}
			|f(x)-S_{n+1}(x)| \le \frac{2}{\pi}\frac{\|f^{(n/2)}\|_{L^1}}{(n/2)^{n/2}}
		\end{align}
		for any even integer $n\ge6$.
	\end{lemma}

	This result implies an exponential convergence rate in terms of $n$.



\begin{thebibliography}{10}
\bibitem{an2021time}
Dong An, Di~Fang, and Lin Lin, \emph{Time-dependent {H}amiltonian simulation of
  highly oscillatory dynamics}, 2021,
  \mbox{\href{http://arxiv.org/abs/arXiv:2111.03103}{arXiv:2111.03103}}.

\bibitem{vanApeldoorn2020optimization}
Joran~van Apeldoorn, Andr{\'a}s Gily{\'e}n, Sander Gribling, and Ronald
  de~Wolf, \emph{Convex optimization using quantum oracles}, Quantum \textbf{4}
  (2020), 220,
  \mbox{\href{http://arxiv.org/abs/arXiv:1809.00643}{arXiv:1809.00643}}
  \url{https://doi.org/10.22331/q-2020-01-13-220}
  .


\bibitem{aspuru2005simulated}
Al{\'a}n Aspuru-Guzik, Anthony~D. Dutoi, Peter~J. Love, and Martin Head-Gordon,
  \emph{Simulated quantum computation of molecular energies}, Science
  \textbf{309} (2005), no.~5741, 1704--1707,
  \href{https://arxiv.org/abs/quant-ph/0604193}{arXiv:quant-ph/0604193}
  \url{https://doi.org/10.1126/science.1113479}.

\bibitem{babbush2016exponentially}
Ryan Babbush, Dominic~W. Berry, Ian~D. Kivlichan, Annie~Y. Wei, Peter~J. Love,
  and Al{\'a}n Aspuru-Guzik, \emph{Exponentially more precise quantum
  simulation of fermions in second quantization}, New Journal of Physics
  \textbf{18} (2016), no.~3, 033032,
  \mbox{\href{http://arxiv.org/abs/arXiv:1506.01020}{arXiv:1506.01020}}
  \url{https://dx.doi.org/10.1088/1367-2630/18/3/033032}
  .

\bibitem{babbush2019quantum}
Ryan Babbush, Dominic~W. Berry, Jarrod~R. McClean, and Hartmut Neven,
  \emph{Quantum simulation of chemistry with sublinear scaling in basis size},
  Npj Quantum Information \textbf{5} (2019), no.~1, 1--7,
  \mbox{\href{http://arxiv.org/abs/arXiv:1807.09802}{arXiv:1807.09802}}
  \url{https://doi.org/10.1038/s41534-019-0199-y}.

\bibitem{babbush2017exponentially}
Ryan Babbush, Dominic~W. Berry, Yuval~R. Sanders, Ian~D. Kivlichan, Artur
  Scherer, Annie~Y. Wei, Peter~J. Love, and Al{\'a}n Aspuru-Guzik,
  \emph{Exponentially more precise quantum simulation of fermions in the
  configuration interaction representation}, Quantum Science and Technology
  \textbf{3} (2017), no.~1, 015006,
  \mbox{\href{http://arxiv.org/abs/arXiv:1506.01029}{arXiv:1506.01029}}
  \url{https://dx.doi.org/10.1088/2058-9565/aa9463}~.

\bibitem{babbush2015chemical}
Ryan Babbush, Jarrod McClean, Dave Wecker, Al{\'a}n Aspuru-Guzik, and Nathan
  Wiebe, \emph{Chemical basis of {T}rotter-{S}uzuki errors in quantum chemistry
  simulation}, Physical Review A \textbf{91} (2015), no.~2, 022311,
  \href{https://arxiv.org/abs/1410.8159}{arXiv:1410.8159}
  \url{https://doi.org/10.1103/PhysRevA.91.022311}.

\bibitem{babbush2017low}
Ryan Babbush, Nathan Wiebe, Jarrod McClean, James McClain, Hartmut Neven, and
  Garnet Kin-Lic Chan, \emph{Low-depth quantum simulation of materials},
  Physical Review X \textbf{8} (2018), no.~1, 011044,
  \mbox{\href{http://arxiv.org/abs/arXiv:1706.00023}{arXiv:1706.00023}}
  \url{https://doi.org/10.1103/PhysRevX.8.011044}.

\bibitem{barnes1986hierarchical}
Josh Barnes and Piet Hut, \emph{A hierarchical ${O}(n \log n)$
  force-calculation algorithm}, nature \textbf{324} (1986), no.~6096, 446--449
  \url{https://doi.org/10.1038/324446a0}.

\bibitem{bauer2020quantum}
Bela Bauer, Sergey Bravyi, Mario Motta, and Garnet Kin-Lic Chan, \emph{Quantum
  algorithms for quantum chemistry and quantum materials science}, Chemical
  Reviews \textbf{120} (2020), no.~22, 12685--12717,
  \href{https://arxiv.org/abs/2001.03685}{arXiv:2001.03685}
  \url{https://doi.org/10.1021/acs.chemrev.9b00829}.

\bibitem{Bea13}
Robert Beals, Stephen Brierley, Oliver Gray, Aram~W. Harrow, Samuel Kutin, Noah
  Linden, Dan Shepherd, and Mark Stather, \emph{Efficient distributed quantum
  computing}, Proceedings of the Royal Society A \textbf{469} (2013), no.~2153,
  20120686,
  \mbox{\href{http://arxiv.org/abs/arXiv:1207.2307}{arXiv:1207.2307}}
  \url{https://doi.org/10.1098/rspa.2012.0686}.

\bibitem{berry2007efficient}
Dominic~W. Berry, Graeme Ahokas, Richard Cleve, and Barry~C. Sanders,
  \emph{Efficient quantum algorithms for simulating sparse {H}amiltonians},
  Communications in Mathematical Physics \textbf{270} (2007), 359--371,
  \href{https://arxiv.org/abs/quant-ph/0508139}{arXiv:quant-ph/0508139}
  \url{https://doi.org/10.1007/s00220-006-0150-x}.

\bibitem{berry2015simulating}
Dominic~W. Berry, Andrew~M. Childs, Richard Cleve, Robin Kothari, and Rolando~D
  Somma, \emph{Simulating {H}amiltonian dynamics with a truncated {T}aylor
  series}, Physical Review Letters \textbf{114} (2015), no.~9, 090502,
  \mbox{\href{http://arxiv.org/abs/arXiv:1412.4687}{arXiv:1412.4687}}
  \url{https://doi.org/10.1103/PhysRevLett.114.090502}.

\bibitem{berry2020time}
Dominic~W. Berry, Andrew~M. Childs, Yuan Su, Xin Wang, and Nathan Wiebe,
  \emph{Time-dependent {H}amiltonian simulation with ${L}^{1}$-norm scaling},
  Quantum \textbf{4} (2020), 254,
  \mbox{\href{http://arxiv.org/abs/arXiv:1906.07115}{arXiv:1906.07115}}
  \url{https://doi.org/10.22331/q-2020-04-20-254}.

\bibitem{berry2019qubitization}
Dominic~W. Berry, Craig Gidney, Mario Motta, Jarrod~R. McClean, and Ryan
  Babbush, \emph{Qubitization of arbitrary basis quantum chemistry leveraging
  sparsity and low rank factorization}, Quantum \textbf{3} (2019), 208,
  \href{https://arxiv.org/abs/1902.02134}{arXiv:1902.02134}
  \url{https://doi.org/10.22331/q-2019-12-02-208}.

\bibitem{bourgain1999growth}
Jean Bourgain, \emph{On growth of {S}obolev norms in linear {S}chr{\"o}dinger
  equations with smooth time dependent potential}, Journal d’Analyse
  Math{\'e}matique \textbf{77} (1999), no.~1, 315--348
  \url{https://doi.org/10.1007/BF02791265}.

\bibitem{boyd2001chebyshev}
John~P. Boyd, \emph{Chebyshev and {F}ourier spectral methods}, Courier
  Corporation, 2001.

\bibitem{brenner2008mathematical}
Susanne~C. Brenner and L.~Ridgway Scott, \emph{The mathematical theory of
  finite element methods}, vol.~3, Springer, 2008
  \url{https://doi.org/10.1007/978-0-387-75934-0}.

\bibitem{campbell2019random}
Earl Campbell, \emph{Random compiler for fast {H}amiltonian simulation},
  Physical Review Letters \textbf{123} (2019), no.~7, 070503,
  \href{https://arxiv.org/abs/1811.08017}{arXiv:1811.08017}
  \url{https://doi.org/10.1103/PhysRevLett.123.070503}.

\bibitem{cao2019quantum}
Yudong Cao, Jonathan Romero, Jonathan~P. Olson, Matthias Degroote, Peter~D.
  Johnson, M{\'a}ria Kieferov{\'a}, Ian~D. Kivlichan, Tim Menke, Borja
  Peropadre, Nicolas P.~D. Sawaya, et~al., \emph{Quantum chemistry in the age
  of quantum computing}, Chemical Reviews \textbf{119} (2019), no.~19,
  10856--10915, \href{https://arxiv.org/abs/1812.09976}{arXiv:1812.09976}
  \url{https://doi.org/10.1021/acs.chemrev.8b00803}.

\bibitem{chakrabarti2020optimization}
Shouvanik Chakrabarti, Andrew~M. Childs, Tongyang Li, and Xiaodi Wu,
  \emph{Quantum algorithms and lower bounds for convex optimization}, Quantum
  \textbf{4} (2020), 221,
  \mbox{\href{http://arxiv.org/abs/arXiv:1809.01731}{arXiv:1809.01731}}
  \url{https://doi.org/10.22331/q-2020-01-13-221}.

\bibitem{childs2004quantum}
Andrew~M. Childs, \emph{Quantum information processing in continuous time},
  Ph.D. thesis, Massachusetts Institute of Technology, 2004.

\bibitem{childs2009limitations}
Andrew~M. Childs and Robin Kothari, \emph{Limitations on the simulation of
  non-sparse {H}amiltonians}, Quantum Information \& Computation \textbf{10}
  (2010), no.~7, 669--684,
  \mbox{\href{http://arxiv.org/abs/arXiv:0908.4398}{arXiv:0908.4398}}
  \url{https://doi.org/10.26421/QIC10.7-8-7}.

\bibitem{childs2020high}
Andrew~M. Childs, Jin-Peng Liu, and Aaron Ostrander, \emph{High-precision
  quantum algorithms for partial differential equations}, Quantum \textbf{5}
  (2021), 574,
  \mbox{\href{http://arxiv.org/abs/arXiv:2002.07868}{arXiv:2002.07868}}
  \url{https://doi.org/10.22331/q-2021-11-10-574}.

\bibitem{childs2018toward}
Andrew~M. Childs, Dmitri Maslov, Yunseong Nam, Neil~J. Ross, and Yuan Su,
  \emph{Toward the first quantum simulation with quantum speedup}, Proceedings
  of the National Academy of Sciences \textbf{115} (2018), no.~38, 9456--9461,
  \mbox{\href{http://arxiv.org/abs/arXiv:1711.10980}{arXiv:1711.10980}}
  \url{https://doi.org/10.1073/pnas.1801723115}.

\bibitem{childs2019theory}
Andrew~M. Childs, Yuan Su, Minh~C. Tran, Nathan Wiebe, and Shuchen Zhu,
  \emph{Theory of {T}rotter error with commutator scaling}, Physical Review X
  \textbf{11} (2021), no.~1, 011020,
  \mbox{\href{http://arxiv.org/abs/arXiv:1912.08854}{arXiv:1912.08854}}
  \url{https://doi.org/10.1103/PhysRevX.11.011020}.

\bibitem{childs2012hamiltonian}
Andrew~M. Childs and Nathan Wiebe, \emph{{H}amiltonian simulation using linear
  combinations of unitary operations}, Quantum Information \& Computation
  \textbf{12} (2012), no.~11-12, 901--924,
  \mbox{\href{http://arxiv.org/abs/arXiv:1202.5822}{arXiv:1202.5822}}
  \url{https://doi.org/10.26421/QIC12.11-12-1}.

\bibitem{dauphin2014identifying}
Yann~N. Dauphin, Razvan Pascanu, Caglar Gulcehre, Kyunghyun Cho, Surya Ganguli,
  and Yoshua Bengio, \emph{Identifying and attacking the saddle point problem
  in high-dimensional non-convex optimization}, Advances in Neural Information
  Processing Systems, pp.~2933--2941, 2014,
  \mbox{\href{http://arxiv.org/abs/arXiv:1406.2572}{arXiv:1406.2572}}.

\bibitem{feynman1982simulating}
Richard~P. Feynman, \emph{Simulating physics with computers}, International
  Journal of Theoretical Physics \textbf{21} (1982), no.~6, 467--488
  \url{https://doi.org/10.1007/BF02650179}.

\bibitem{fyodorov2007replica}
Yan~V. Fyodorov and Ian Williams, \emph{Replica symmetry breaking condition
  exposed by random matrix calculation of landscape complexity}, Journal of
  Statistical Physics \textbf{129} (2007), no.~5-6, 1081--1116,
  \mbox{\href{http://arxiv.org/abs/arXiv:cond-mat/0702601}{arXiv:cond-mat/0702601}}
  \url{https://doi.org/10.1007/s10955-007-9386-x}.

\bibitem{gilyen2019quantum}
Andr{\'a}s Gily{\'e}n, Yuan Su, Guang~Hao Low, and Nathan Wiebe, \emph{Quantum
  singular value transformation and beyond: exponential improvements for
  quantum matrix arithmetics}, Proceedings of the 51st Annual ACM SIGACT
  Symposium on Theory of Computing, pp.~193--204, 2019,
  \mbox{\href{http://arxiv.org/abs/arXiv:1806.01838}{arXiv:1806.01838}}
  \url{https://doi.org/10.1145/3313276.3316366}.

\bibitem{giuliani2005quantum}
Gabriele Giuliani and Giovanni Vignale, \emph{Quantum theory of the electron
  liquid}, Cambridge University Press, 2005
  \url{https://doi.org/10.1017/CBO9780511619915}.

\bibitem{greengard1987fast}
Leslie Greengard and Vladimir Rokhlin, \emph{A fast algorithm for particle
  simulations}, Journal of Computational Physics \textbf{73} (1987), no.~2,
  325--348
  \url{https://doi.org/10.1016/0021-9991(87)90140-9}.

\bibitem{haah2018quantum}
Jeongwan Haah, Matthew Hastings, Robin Kothari, and Guang~Hao Low,
  \emph{Quantum algorithm for simulating real time evolution of lattice
  {H}amiltonians}, Proceedings of the 59th Annual Symposium on Foundations of
  Computer Science, pp.~350--360, IEEE, 2018,
  \mbox{\href{http://arxiv.org/abs/arXiv:1801.03922}{arXiv:1801.03922}}
  \url{https://doi.org/10.1137/18M1231511}.

\bibitem{hastings2014improving}
Matthew~B. Hastings, Dave Wecker, Bela Bauer, and Matthias Troyer,
  \emph{Improving quantum algorithms for quantum chemistry}, Quantum
  Information \& Computation \textbf{15} (2015), no.~1-2, 1--21,
  \href{https://arxiv.org/abs/1403.1539}{arXiv:1403.1539}
  \url{ https://doi.org/10.26421/QIC15.1-2-1}.

\bibitem{hildebrand1987introduction}
Francis~Begnaud Hildebrand, \emph{Introduction to numerical analysis}, Courier
  Corporation, 1987
  \url{https://doi.org/10.1007/978-0-387-21738-3}.

\bibitem{jin2018accelerated}
Chi Jin, Praneeth Netrapalli, and Michael~I. Jordan, \emph{Accelerated gradient
  descent escapes saddle points faster than gradient descent}, Conference on
  Learning Theory, pp.~1042--1085, 2018,
  \mbox{\href{http://arxiv.org/abs/arXiv:1711.10456}{arXiv:1711.10456}}.

\bibitem{jin2021quantum}
Shi Jin, Xiantao Li, and Nana Liu, \emph{Quantum simulation in the semi-classical regime}, Quantum \textbf{6} (2022), 739
  \mbox{\href{http://arxiv.org/abs/arXiv:2112.13279}{arXiv:2112.13279}}
  \url{https://doi.org/10.22331/q-2022-06-17-739}.

\bibitem{jordan2005fast}
Stephen~P. Jordan, \emph{Fast quantum algorithm for numerical gradient
  estimation}, Physical Review Letters \textbf{95} (2005), no.~5, 050501,
  \mbox{\href{http://arxiv.org/abs/arXiv:quant-ph/0405146}{arXiv:quant-ph/0405146}}
  \url{https://doi.org/10.1103/PhysRevLett.95.050501}.

\bibitem{jordan2012quantum}
Stephen~P. Jordan, Keith~S.M. Lee, and John Preskill, \emph{Quantum algorithms
  for quantum field theories}, Science \textbf{336} (2012), no.~6085,
  1130--1133,
  \mbox{\href{http://arxiv.org/abs/arXiv:1111.3633}{arXiv:1111.3633}}
  \url{https://doi.org/10.1126/science.1217069}.

\bibitem{kassal2008polynomial}
Ivan Kassal, Stephen~P. Jordan, Peter~J. Love, Masoud Mohseni, and Al{\'a}n
  Aspuru-Guzik, \emph{Polynomial-time quantum algorithm for the simulation of
  chemical dynamics}, Proceedings of the National Academy of Sciences
  \textbf{105} (2008), no.~48, 18681--18686,
  \href{https://arxiv.org/abs/0801.2986}{arXiv:0801.2986}
  \url{https://doi.org/10.1073/pnas.0808245105}.

\bibitem{kivlichan2017bounding}
Ian~D. Kivlichan, Nathan Wiebe, Ryan Babbush, and Al{\'a}n Aspuru-Guzik,
  \emph{Bounding the costs of quantum simulation of many-body physics in real
  space}, Journal of Physics A: Mathematical and Theoretical \textbf{50}
  (2017), no.~30, 305301,
  \mbox{\href{http://arxiv.org/abs/arXiv:1608.05696}{arXiv:1608.05696}}
  \url{https://dx.doi.org/10.1088/1751-8121/aa77b8}.

\bibitem{lee2020even}
Joonho Lee, Dominic Berry, Craig Gidney, William~J. Huggins, Jarrod~R. McClean,
  Nathan Wiebe, and Ryan Babbush, \emph{Even more efficient quantum
  computations of chemistry through tensor hypercontraction}, PRX Quantum
  \textbf{2} (2021), no.~3, 030305,
  \href{https://arxiv.org/abs/2011.03494}{arXiv:2011.03494}
  \url{https://doi.org/10.1103/PRXQuantum.2.030305}.

\bibitem{lloyd1996universal}
Seth Lloyd, \emph{Universal quantum simulators}, Science (1996), 1073--1078
\url{https://doi.org/10.1126/science.273.5278.1073}.

\bibitem{low2019qubitization}
Guang~Hao Low and Isaac~L. Chuang, \emph{{H}amiltonian simulation by
  qubitization}, Quantum \textbf{3} (2019), 163,
  \mbox{\href{http://arxiv.org/abs/arXiv:1610.06546}{arXiv:1610.06546}}
  \url{https://doi.org/10.22331/q-2019-07-12-163}.

\bibitem{low2018hamiltonian}
Guang~Hao Low and Nathan Wiebe, \emph{{H}amiltonian simulation in the
  interaction picture}, 2018,
  \mbox{\href{http://arxiv.org/abs/arXiv:1805.00675}{arXiv:1805.00675}}.

\bibitem{martin2004electronic}
Richard~M. Martin, \emph{Electronic structure}, Cambridge University Press,
  2004
  \url{https://doi.org/10.1017/CBO9780511805769}.

\bibitem{mcardle2021exploiting}
Sam McArdle, Earl Campbell, and Yuan Su, \emph{Exploiting fermion number in
  factorized decompositions of the electronic structure {H}amiltonian},
  Physical Review A \textbf{105} (2022), no.~1, 012403,
  \mbox{\href{http://arxiv.org/abs/arXiv:2107.07238}{arXiv:2107.07238}}
  \url{https://doi.org/10.1103/PhysRevA.105.012403}.

\bibitem{mcclean2014exploiting}
Jarrod~R. McClean, Ryan Babbush, Peter~J. Love, and Al{\'a}n Aspuru-Guzik,
  \emph{Exploiting locality in quantum computation for quantum chemistry}, The
  Journal of Physical Chemistry Letters \textbf{5} (2014), no.~24, 4368--4380
  \url{https://doi.org/10.1021/jz501649m}.

\bibitem{motta2018low}
Mario Motta, Erika Ye, Jarrod~R. McClean, Zhendong Li, Austin~J. Minnich, Ryan
  Babbush, and Garnet Kin-Lic Chan, \emph{Low rank representations for quantum
  simulation of electronic structure}, npj Quantum Information \textbf{7}
  (2021), no.~1, 1--7,
  \mbox{\href{http://arxiv.org/abs/arXiv:1808.02625}{arXiv:1808.02625}}
  \url{https://doi.org/10.1038/s41534-021-00416-z}.

\bibitem{poulin2014trotter}
David Poulin, Matthew~B. Hastings, David Wecker, Nathan Wiebe, Andrew~C.
  Doberty, and Matthias Troyer, \emph{The {T}rotter step size required for
  accurate quantum simulation of quantum chemistry}, Quantum Information \&
  Computation \textbf{15} (2015), no.~5-6, 361--384,
  \href{https://arxiv.org/abs/1406.4920}{arXiv:1406.4920}
  \url{https://doi.org/10.26421/QIC15.5-6-1}.

\bibitem{preskill2019simulating}
John Preskill, \emph{Simulating quantum field theory with a quantum computer},
  The 36th Annual International Symposium on Lattice Field Theory, vol. 334,
  p.~024, SISSA Medialab, 2019,
  \href{https://arxiv.org/abs/1811.10085}{arXiv:1811.10085}
  \url{DOI: https://doi.org/10.22323/1.334.0024}.

\bibitem{reiher2017elucidating}
Markus Reiher, Nathan Wiebe, Krysta~M. Svore, Dave Wecker, and Matthias Troyer,
  \emph{Elucidating reaction mechanisms on quantum computers}, Proceedings of
  the National Academy of Sciences \textbf{114} (2017), no.~29, 7555--7560,
  \href{https://arxiv.org/abs/1605.03590}{arXiv:1605.03590}
  \url{https://doi.org/10.1073/pnas.1619152114}.

\bibitem{sarin1998analyzing}
Vivek Sarin, Ananth Grama, and Ahmed Sameh, \emph{Analyzing the error bounds of
  multipole-based treecodes}, SC'98: Proceedings of the 1998 ACM/IEEE
  Conference on Supercomputing, pp.~19--19, IEEE, 1998
  \url{https://doi.org/10.1109/SC.1998.10041}.

\bibitem{seeley2012bravyi}
Jacob~T. Seeley, Martin~J. Richard, and Peter~J. Love, \emph{The
  {B}ravyi-{K}itaev transformation for quantum computation of electronic
  structure}, The Journal of Chemical Physics \textbf{137} (2012), no.~22,
  224109, \href{https://arxiv.org/abs/1208.5986}{arXiv:1208.5986}
  \url{https://doi.org/10.1063/1.4768229}.

\bibitem{tang2006spectral}
Jie Shen and Tao Tang, \emph{Spectral and high-order methods with
  applications}, Science Press Beijing, 2006,
  \href{https://www.math.purdue.edu/~shen7/sp_intro12/book.pdf}{https://www.math.purdue.edu/~shen7/sp{\_}intro12/book.pdf}.

\bibitem{shi2020learning}
Bin Shi, Weijie~J. Su, and Michael~I. Jordan, \emph{On learning rates and
  {S}chr{\"o}dinger operators}, 2020,
  \mbox{\href{http://arxiv.org/abs/arXiv:2004.06977}{arXiv:2004.06977}}.

\bibitem{su2021fault}
Yuan Su, Dominic~W Berry, Nathan Wiebe, Nicholas Rubin, and Ryan Babbush,
  \emph{Fault-tolerant quantum simulations of chemistry in first quantization},
  PRX Quantum \textbf{2} (2021), no.~4, 040332,
  \href{https://arxiv.org/abs/2105.12767}{arXiv:2105.12767}
  \url{https://doi.org/10.1103/PRXQuantum.2.040332}.

\bibitem{su2020nearly}
Yuan Su, Hsin-Yuan Huang, and Earl~T. Campbell, \emph{Nearly tight
  {T}rotterization of interacting electrons}, Quantum \textbf{5} (2021), 495,
  \href{https://arxiv.org/abs/2012.09194}{arXiv:2012.09194}
  \url{https://doi.org/10.22331/q-2021-07-05-495}.

\bibitem{suzuki1991general}
Masuo Suzuki, \emph{General theory of fractal path integrals with applications
  to many-body theories and statistical physics}, Journal of Mathematical
  Physics \textbf{32} (1991), no.~2, 400--407
  \url{https://doi.org/10.1063/1.529425}.

\bibitem{szabo1991finite}
Barna Szab{\'o} and Ivo Babu{\v{s}}ka, \emph{Finite element analysis}, John
  Wiley \& Sons, 1991.

\bibitem{toloui2013quantum}
Borzu Toloui and Peter~J. Love, \emph{Quantum algorithms for quantum chemistry
  based on the sparsity of the {CI}-matrix}, 2013,
  \href{https://arxiv.org/abs/1312.2579}{arXiv:1312.2579}.

\bibitem{von2020quantum}
Vera von Burg, Guang~Hao Low, Thomas H{\"a}ner, Damian~S. Steiger, Markus
  Reiher, Martin Roetteler, and Matthias Troyer, \emph{Quantum computing
  enhanced computational catalysis}, Physical Review Research \textbf{3}
  (2021), no.~3, 033055,
  \mbox{\href{http://arxiv.org/abs/arXiv:2007.14460}{arXiv:2007.14460}}
  \url{https://doi.org/10.1103/PhysRevResearch.3.033055}.

\bibitem{wecker2014gate}
Dave Wecker, Bela Bauer, Bryan~K. Clark, Matthew~B. Hastings, and Matthias
  Troyer, \emph{Gate-count estimates for performing quantum chemistry on small
  quantum computers}, Physical Review A \textbf{90} (2014), no.~2, 022305,
  \href{https://arxiv.org/abs/1312.1695}{arXiv:1312.1695}
  \url{https://doi.org/10.1103/PhysRevA.90.022305}.

\bibitem{whitfield2011simulation}
James~D. Whitfield, Jacob Biamonte, and Al{\'a}n Aspuru-Guzik, \emph{Simulation
  of electronic structure {H}amiltonians using quantum computers}, Molecular
  Physics \textbf{109} (2011), no.~5, 735--750,
  \href{https://arxiv.org/abs/1001.3855}{arXiv:1001.3855}
  \url{https://doi.org/10.1080/00268976.2011.552441}.

\bibitem{wiesner1996simulations}
Stephen Wiesner, \emph{Simulations of many-body quantum systems by a quantum
  computer}, 1996,
  \href{https://arxiv.org/abs/quant-ph/9603028}{arXiv:quant-ph/9603028}.

\bibitem{zalka1998efficient}
Christof Zalka, \emph{Efficient simulation of quantum systems by quantum
  computers}, Fortschritte der Physik: Progress of Physics \textbf{46} (1998),
  no.~6-8, 877--879,
  \href{https://arxiv.org/abs/quant-ph/9603026}{arXiv:quant-ph/9603026}.

\bibitem{zhang2020quantum}
Chenyi Zhang, Jiaqi Leng, and Tongyang Li, \emph{Quantum algorithms for
  escaping from saddle points}, Quantum \textbf{5} (2021), 529,
  \mbox{\href{http://arxiv.org/abs/arXiv:2007.10253v3}{arXiv:2007.10253v3}}
  \url{https://doi.org/10.22331/q-2021-08-20-529}.

\bibitem{zhang2021escape}
Chenyi Zhang and Tongyang Li, \emph{Escape saddle points by a simple
  gradient-descent based algorithm}, Advances in Neural Information Processing
  Systems, vol.~34, 2021,
  \mbox{\href{http://arxiv.org/abs/arXiv:2111.14069}{arXiv:2111.14069}}.
\end{thebibliography}
\end{document}